\documentclass[journal]{IEEEtran}
\usepackage{cite}
\usepackage{array}
\usepackage{graphicx}
\usepackage{amsmath,amssymb,amsfonts}
\usepackage[center]{caption}
\usepackage[font=footnotesize]{caption}
\usepackage{xcolor}
\usepackage{algorithm,algorithmicx,algpseudocode}
\usepackage{bm}
\usepackage{mathtools}
\usepackage{nccmath}
\definecolor{Green}{rgb}{0.13, 0.55, 0.13}
\usepackage{comment}
\excludecomment{draftnotes}
\usepackage{subcaption}
\usepackage{enumitem}

\usepackage{makecell}
\usepackage{mathrsfs}
\usepackage{upgreek}
\usepackage{pifont}  
\newcommand{\cmark}{\ding{51}} 
\newcommand{\xmark}{\ding{55}} 

\usepackage{multirow}
\usepackage{tabularx}
\usepackage{booktabs}

\usepackage{titlesec}

\usepackage{amsthm}
\newtheorem{theorem}{Theorem}
\newtheorem{proposition}{Proposition}
\newtheorem{lemma}{Lemma}

\DeclareMathOperator*{\argmin}{arg\,min}

\newcommand{\algrule}[1][.2pt]{\par\vskip.5\baselineskip\hrule height #1\par\vskip.5\baselineskip}

\title{Version-Aware Communication in Multi-Hop IoT Networks with Feedback}
\author{\IEEEauthorblockN{Erfan Delfani and Nikolaos Pappas, \IEEEmembership{Senior Member, IEEE}}
        \thanks{The authors are with the Department of Computer and Information Science at Linköping University, Sweden, email: \{\textit{erfan.delfani, nikolaos.pappas\}@liu.se}. This work has been supported in part by the Swedish Research Council (VR), ELLIIT, and the European Union (ELIXIRION, 101120135, 6G-LEADER, 101192080, and SOVEREIGN, 101131481). A shorter version has been published in \cite{delfani2025timestamps}.}}
\date{August 2025}

\begin{document}

\setlength{\abovecaptionskip}{2pt}
\setlength{\belowcaptionskip}{-4pt}

\maketitle

\begin{abstract}
		Timely communication of information in Internet of Things (IoT) networks is critical to enhancing system performance and energy efficiency by minimizing the transmission of outdated or redundant data. Although timeliness metrics such as the Age of Information (\textsc{AoI}) effectively quantify information freshness, they do not account for content evolution. The Version Age of Information (\textsc{VAoI}) addresses this gap by tracking version lag at the receiver, thereby providing a practical content-aware metric. However, prior research has primarily focused on first-moment analyses in single-hop settings, leaving the distributional properties of \textsc{VAoI} in multi-hop networks, as well as the impact of feedback mechanisms, unexplored. In this study, we provide a comprehensive characterization of \textsc{VAoI} in multi-hop networks with transmission constraints and acknowledgment-based feedback. A bi-level optimization framework is formulated to jointly optimize the update policy of a rate-constrained source and the feedback-aware forwarding policies of the intermediate nodes, aiming to minimize communication overhead while maintaining \textsc{VAoI} performance at the destination. We show that the optimal source policy follows a threshold-based update strategy and derive the optimal threshold in closed form. For both the optimal threshold policy and a randomized baseline, we obtain closed-form expressions for the stationary distribution and average \textsc{VAoI}, along with the corresponding update rates across network nodes under feedback-aware forwarding. Numerical results corroborate the analytical findings and illustrate the advantages of utilizing \textsc{VAoI} and feedback to reduce redundant transmissions while preserving data freshness and informativeness in multi-hop systems. 
	\end{abstract}

\section{Introduction}
Efficient data management is critical for ensuring optimal performance in communication networks across diverse scenarios, ranging from single-hop IoT monitoring systems to multi-hop satellite communication networks. As the volume of data generated by network nodes increases, transmitting all data indiscriminately without considering its semantic significance or task relevance becomes increasingly impractical and unsustainable. Such an approach leads to excessive consumption of critical resources, including energy and bandwidth, ultimately degrading system performance. To address these challenges, there is an urgent need for network management approaches that dynamically optimize data transmission by leveraging semantic metrics to deliver the most \emph{timely} and \emph{informative} data within a constrained transmission frequency~\cite{kountouris2021semantics, luo2025survey}. This requires a shift from traditional passive, semantics-agnostic communication toward active, goal-oriented, semantics-aware approaches.

The Age of Information (\textsc{AoI})~\cite{kaul2012real} is a widely used metric that quantifies the freshness or timeliness of information in status update systems, defined as the time elapsed since the generation of the most recently received data. AoI-aware scheduling minimizes staleness by adapting transmissions to source data arrivals and network service times~\cite{yates2021age}. However, \textsc{AoI} captures freshness solely through data \emph{timestamps}, without accounting for actual changes in the source content. As a result, simply \emph{refreshing timestamps} may fail to deliver new information, and avoiding such redundant updates can reduce transmission and energy consumption.

To address this limitation, \emph{content-based} metrics such as Age of Incorrect Information (\textsc{AoII}) \cite{maatouk2020age} and Version Age of Information (\textsc{VAoI}) \cite{yates2021vage} have been introduced. \textsc{AoII} adds a distortion-aware dimension by measuring the staleness of incorrect information, specifically when the receiver's content deviates from the source, unlike \textsc{AoI}, which treats correct and incorrect data uniformly. However, \textsc{AoII} requires precise knowledge of the information \emph{state} at both the source and destination for comparison, which is practical only when the state space is small, fully modeled, and all state transitions are known. In many real-world applications, such complete knowledge may not be available.

In such cases, \textsc{VAoI} provides a more practical content-based metric by focusing solely on \emph{content changes} at the source, where data evolve through successive, non-reverting \emph{versions}. This requires only minimal knowledge: at any given time, either a new or a previous version exists, and the receiver must track these versions as timely as possible. Defined as the number of versions by which the receiver lags behind the source, \textsc{VAoI} further improves upon \textsc{AoI} by replacing timestamps with version numbers, thereby eliminating the challenging requirement of clock synchronization between the transmitter and receiver \cite{Mehrdad2025CL}. It is computed simply by comparing the receiver's stored version with the source's current version.

While these metrics have attracted considerable attention, most existing research has focused primarily on single-hop communication setups and on first-moment analyses, i.e., average values. In contrast, there remains a notable gap regarding multi-hop communication scenarios and the analysis of the full distributions of these metrics, particularly \textsc{VAoI}. Such analyses are essential for a deeper understanding of system behavior and for enabling more effective optimization, especially in resource-constrained IoT systems.

Another mechanism for improving the efficiency of data communication networks is the provision of reliable and prompt \emph{acknowledgment feedback} upon data delivery. Feedback-enabled networks facilitate active, closed-loop decision-making at transmitting nodes, enabling more efficient data transmission and the delivery of the most timely and informative data \cite{feng2022timely,munari2025s}. This feedback mechanism is particularly important in multi-hop networks. Whereas a lack of feedback in single-hop links affects only one transmitter, in multi-hop networks it causes a cascade of redundant retransmissions across intermediate nodes, wasting network resources without providing any benefit. Investigating the impact of feedback on communication overhead in multi-hop networks while preserving information semantics, specifically \textsc{VAoI}, is another direction that warrants further exploration.

In this work, we address these gaps by considering \textsc{VAoI}-aware and feedback-enabled communication approaches, in contrast to \textsc{VAoI}-agnostic and feedback-disabled approaches, in a multi-hop setup to enable efficient data delivery under source transmission rate constraints. \textsc{VAoI}-aware communication improves network efficiency by prioritizing the transmission of fresh and informative data. \emph{When network resources are limited, this approach reduces transmissions of stale or redundant information, thereby conserving resources and enabling timely and informative updates. When resource constraints are relaxed, \textsc{VAoI}-aware communication maintains data freshness and informativeness while reducing unnecessary resource consumption by filtering redundant transmissions.}

We consider these advantages in a multi-hop network consisting of a rate-constrained source node (node $0$), a sequence of $N$ intermediate relay nodes (nodes $1,2,\dots,N$) with no strict rate constraints, and a destination node (node $N+1$). The source node determines its transmission or \emph{update policy}, denoted by $\phi$, to optimize the timeliness and informativeness of the data at the destination node, as captured by \textsc{VAoI}. Each intermediate node $i \in \{1,2,\dots,N\}$ adopts a relaying or \emph{forwarding policy} $\theta_i$, whose primary objective is to reliably deliver all received data versions to the destination, thereby optimizing the \textsc{VAoI} at node~$N+1$. As a secondary objective, each forwarding policy also seeks to minimize communication overhead. The resulting problem naturally leads to a bi-level optimization problem:

\begin{itemize}
    \item \textbf{Upper-Level (\textsc{VAoI} Optimization):}  
    The network nodes optimize the \textsc{VAoI} at the destination subject to a transmission \emph{rate constraint} at the source. All intermediate nodes are assumed to employ an ideal \emph{always-update} forwarding policy, thereby establishing a baseline for the \emph{best achievable VAoI} at the destination.

    \item \textbf{Lower-Level (Update Rate Optimization):}  
    Given the optimal source policy $\phi^*$ obtained from the upper-level problem, the objective is to minimize the transmission rates of the intermediate nodes while preserving the best achievable \textsc{VAoI} at the destination. To reduce communication overhead, intermediate nodes use acknowledgment (\textsc{ACK}) feedback from downstream nodes to suppress redundant retransmissions once the latest data version has been successfully delivered, following a \emph{feedback-aware} (or \emph{\textsc{VAoI}-aware}) policy.
\end{itemize}

Thus, the upper-level optimization determines the optimal update policy for the rate-constrained source, e.g., an \emph{IoT device}, while the lower-level optimization determines the minimum update/forwarding rates under a feedback-aware policy at the intermediate nodes, e.g., a \emph{network operator}. The formal problem formulation is presented in Section~\ref{sec_ProblemFormulation}.

\subsection{Main Contributions}
The main contributions of this study are as follows:
\begin{itemize}
    \item We analyze the \textsc{VAoI} evolution at the nodes of a multi-hop network under a general update policy $\phi$ at a rate-constrained source, assuming an always-update forwarding policy at intermediate nodes.
    \item We formulate a bi-level optimization problem that first characterizes the best achievable \textsc{VAoI} under a source-rate constraint with always-update forwarding, and then determines the minimum update rates at intermediate nodes under a feedback-aware forwarding policy while preserving this \textsc{VAoI}.
    \item We show that the optimal source update policy, $\phi^*$, is a threshold policy and derive the optimal threshold in closed form.
    \item We derive closed-form expressions for the stationary distribution and average \textsc{VAoI} at network nodes under two source policies: the optimal (version-aware) threshold policy $\phi^*$ and a randomized stationary policy $\phi_{\mathcal{R}}$, serving as a version-agnostic baseline.
    \item We derive closed-form expressions for the update rates at network nodes under a feedback-aware forwarding policy for both threshold and randomized source policies.
    \item We validate the analytical results via simulations and investigate the behavior of \textsc{VAoI} and update rates in the multi-hop network.
\end{itemize}

\subsection{Related Works}
\label{sec_RelatedWorks}
The literature on information freshness and semantic metrics can be broadly categorized by network topology and depth of analysis. For single-hop setups, existing research has primarily focused on first-moment or average analysis. Several studies have investigated the distributions of \textsc{AoI} and Peak \textsc{AoI} (\textsc{PAoI}) in continuous-time systems using queueing theory \cite{costa2016age,champati2019distribution,inoue2019general,yates2020age,Chiariotti2021PAoI,jiang2021joint,abd2022closed,moltafet2022moment,fiems2023age,akar2025age,inoue2025characterizing}, while others have considered discrete-time settings \cite{kosta2021age,zhang2021age,akar2021discrete,ji2024age,Zhang2025AoIVehicles}. Notably, \cite{kosta2021age} derives general expressions for the stationary distributions and generating functions of \textsc{AoI} and \textsc{PAoI} in discrete-time single-server queues under various disciplines, together with methods for nonlinear age functions. Extending stochastic hybrid system techniques, \cite{zhang2021age} models \textsc{AoI} and packet age as a two-dimensional Markov process in bufferless queues with Bernoulli arrivals. A matrix-analytic quasi-birth--death framework is proposed in \cite{akar2021discrete} to obtain exact per-source \textsc{AoI} and \textsc{PAoI} distributions in multi-source IoT systems with phase-type service times. The study in \cite{ji2024age} investigates age-optimal scheduling with delayed feedback and long-term constraints, providing closed-form benchmarks for random and deterministic policies, while \cite{Zhang2025AoIVehicles} analyzes \textsc{AoI} and \textsc{PAoI} in multi-source Ber/Geo/1/1 systems under preemptive and non-preemptive policies.
	
For content-based metrics, several works have examined the distribution of \textsc{AoII} \cite{maatouk2020age,chen2024minimizing,salimnejad2024age}. Specifically, \cite{maatouk2020age} derives stationary \textsc{AoII} distributions for symmetric multi-state Markov sources under always-update and threshold policies. \cite{chen2024minimizing} studies \textsc{AoII} in slotted systems with random delays for two-state Markov sources under threshold policies. Using \textsc{DTMC} analysis, \cite{salimnejad2024age} derives stationary \textsc{AoII} and \textsc{AoIV} distributions for two-state Markov sources under given transmission policies. \textsc{AoCI} was introduced in \cite{wang2021age}, which derives optimal thresholds that minimize a weighted sum of \textsc{AoCI} and update cost. Stationary \textsc{VAoI} distributions for energy-harvesting systems are modeled in \cite{delfani2024semantics} using \textsc{DTMC}s with stochastic energy arrivals and threshold-based transmissions, while \cite{karevvanavar2024version} analyzes \textsc{VAoI} in NOMA fading broadcast channels with random version arrivals and power constraints. Although these studies provide valuable insights into the behavior of freshness metrics, they are limited to single-hop configurations.
	
In multi-hop networks, most existing studies on two-hop \cite{arafa2019timely,li2020age,gu2021optimizing} and general multi-hop settings \cite{Talak2018Mhop, Bedewy2019Mhop, buyukates2019age,lou2021boosting,liu2021minimizing,ke2025information,Chiariotti2022Mhop, Tripathi2023Mhop, Kaswan2023Mhop, Sinha2024Tandem,asvadi2024age} predominantly focus on analyzing or optimizing average freshness metrics, particularly \textsc{AoI}, while largely neglecting content-aware metrics such as \textsc{VAoI}, as well as the characterization of the full distributional behavior.
A small subset of prior work goes beyond this focus. \cite{ayan2020probability} derives the distribution of discrete-time \textsc{AoI} in $N$-hop systems with time-invariant packet loss via recursive methods, while \cite{Vikhrova2020Mhop} investigates the distributions of \textsc{AoI} and \textsc{PAoI} in continuous-time two-hop networks. However, these studies remain centered on AoI-based metrics and do not extend to content-aware measures such as \textsc{VAoI}. More recently, \cite{Delfani2025LEO} studies \textsc{VAoI} in multi-hop satellite networks under an always-update forwarding policy. Nevertheless, this work does not account for update rates at intermediate relaying nodes, nor does it investigate the role of feedback mechanisms, which are crucial for efficient closed-loop operation and reducing redundant retransmissions and communication overhead in multi-hop topologies.
Overall, the current literature reveals a clear gap in the joint analysis of multi-hop communication, feedback-enabled control, and the full distributional behavior of content-based freshness metrics such as \textsc{VAoI}. Addressing this gap is essential for understanding the cascading effects of feedback in multi-hop systems and for designing resource-efficient policies that preserve information semantics under network constraints. Unlike the aforementioned works, this paper jointly addresses multi-hop VAoI optimization, stationary distribution analysis, and feedback-aware forwarding under source transmission constraints.

\section{System Model and Formulation}

We consider a multi-hop line network designed for the communication of status updates from a source node to a destination through a sequence of intermediate relay nodes. The network consists of $N+2$ nodes indexed by $i \in \{0,1,\dots,N+1\}$, where node $0$ represents the source and node $N+1$ denotes the final destination. Nodes $1,\ldots,N$ act as intermediate forwarding nodes that relay updates toward the destination over unreliable communication links.
A representative real-world scenario is illustrated in Fig.~\ref{fig_SysModel}, where a remote energy-limited IoT device (node $0$) samples updates from an information source, such as a physical process, and transmits the resulting update packets to a distant ground node (node $N+1$). The updates are delivered through $N$-hop communication via a sequence of Low Earth Orbit (LEO) satellites that serve as intermediate relay nodes.
In this model, the source node is assumed to be energy-constrained and, therefore, has limited transmission capability. The intermediate nodes are assumed to have no strict resource constraints and aim to deliver data from the source to the destination as reliably as possible while minimizing communication overhead. In all analyses, we consider a slotted (discrete-time) system indexed by $t=0,1,2,\dots$. The details of the system model are described below.

\begin{figure}[!t]
    \centering
    \includegraphics[trim={4.2cm 4cm 4.2cm 4.2cm}, clip, width=\linewidth]{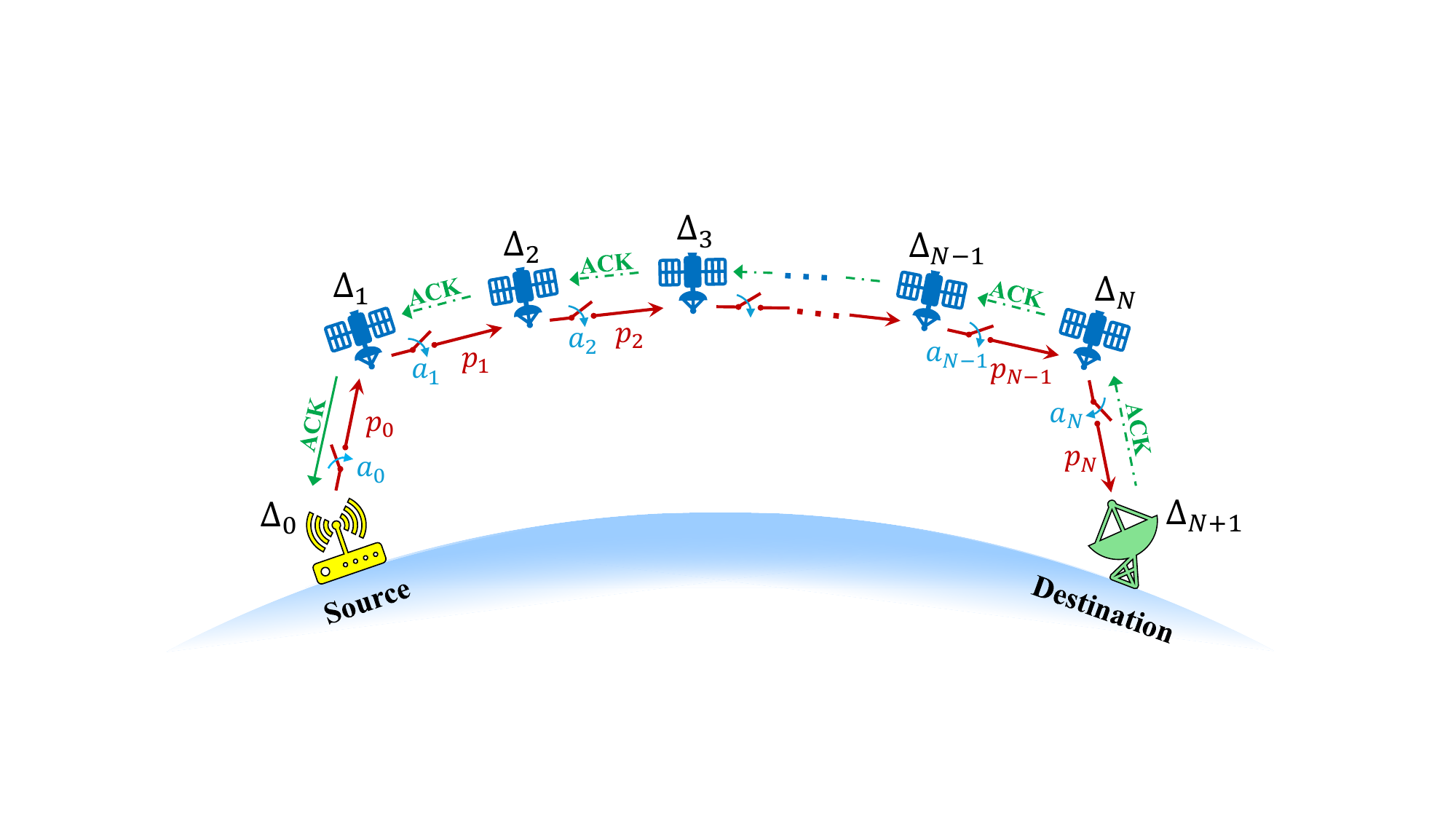}
    \caption{System model of a multi-hop network.}
    \label{fig_SysModel}
\end{figure}

\subsubsection{Communication Links}
\textit{Forward links:}
In this setup, at the beginning of each time slot and in accordance with the update policies outlined in Sec.~\ref{Sec_UpdatePolicies}, each node may attempt to transmit an update to its next-hop node. The completion of each transmission occupies exactly one time slot. Accordingly, updates propagate through the network via hop-by-hop transmissions over unreliable links. 
The communication link between node $i$ and node $i+1$ is modeled as a packet erasure channel. Specifically, when node $i$ attempts a transmission during time slot $t$, the transmission succeeds with probability (w.p.) $p_i$, and consequently, fails (i.e., is dropped) w.p. $\epsilon_i = 1-p_i$. 
We define the channel success indicator as:
\begin{align}
    h_i(t) \sim \text{Bernoulli}(p_i), \quad i \in \{0,1,2,\dots,N\},
\end{align}
where $h_i(t)=1$ indicates that a transmission from node $i$ to node $i+1$ at time slot $t$ is successful, and $h_i(t)=0$ otherwise. Channel outcomes are assumed to be independent across both time slots and different links.
Moreover, each node is assumed to maintain a buffer capable of storing a single update. The buffer always retains the most recently received update (version), discarding any older updates upon the arrival of a new one.

\textit{Backward (feedback) links:}
For the feedback-aware policy, we consider feedback links to acknowledge the successful delivery of updates within the network. In particular, when enabled, node $i+1$ sends an \textsc{ACK} to node $i$ upon receiving an update from that node. We assume that \textsc{ACK}s from receivers to transmitters are instantaneous and error-free, as sufficient resources can be allocated to ensure the reliable delivery of these small \textsc{ACK} packets.

\subsubsection{Update and Version Generation} 
The source node $0$ continuously monitors the information source and generates a status-update packet in each discrete time slot. However, not every update necessarily contains a new version of the information. We assume that the generation of a new version is modeled as an independent and identically distributed (i.i.d.) Bernoulli process with probability $p_g$. This stochastic model effectively captures the evolution of the information source over the long term.
The version generation process at the source is described by:
\begin{align}
w(t) \sim \text{Bernoulli}(p_g),
\end{align}
where $w(t) = 1$ indicates that a new version is generated in time slot $t$, and $w(t) = 0$ otherwise. Given this version generation process, the version index at the source node, denoted by $V_0(t)$, evolves according to $V_0(t+1)=V_0(t)+w(t)$. 

\subsubsection{Temporal Order of Events}
The system follows a temporal ordering within each time slot. Transmission attempts occur at the beginning of the slot after observing the current state of the system, while channel outcomes and packet receptions are determined during the slot (see Fig. \ref{fig_NodesEventOrder}). Moreover, the generation of a new version at the source is completed at the end of the slot (see Fig. \ref{fig_SourceEventOrder}). Under this timing structure, if node $i$ receives a new version during slot $t$, it can start forwarding that version at the beginning of slot $t+1$.
\begin{figure}[!t]
    \centering
    \includegraphics[width=0.95\linewidth]{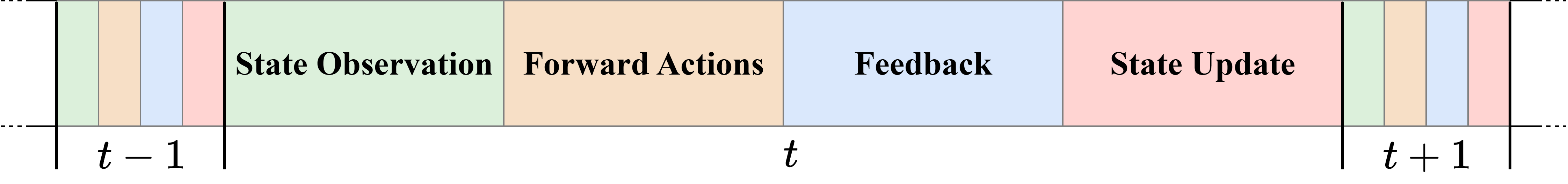}
    \vspace{-4pt}
    \caption{Order of events at the network nodes.}
    \label{fig_NodesEventOrder}
    \bigskip
    \centering
    \includegraphics[width=0.95\linewidth]{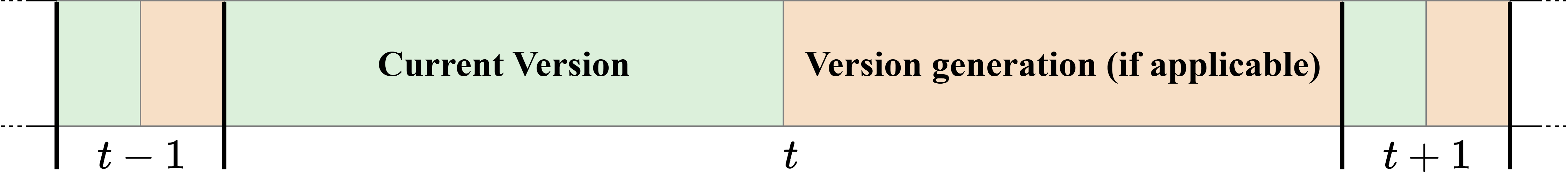}
    \vspace{-4pt}
    \caption{Order of events at the source node.}
    \label{fig_SourceEventOrder}
\end{figure}

\subsection{Update Policies and ACK Feedback}
\label{Sec_UpdatePolicies}

\subsubsection{Update Policy at the Source}
The transmission behavior of the source node is characterized by an update policy $\phi=\{\mathrm{a}_0^\phi(t)\}_{t=0}^{\infty}$, where $\mathrm{a}_0^\phi(t) \in \{0,1\}$. In particular, $\mathrm{a}_0^\phi(t)=1$ indicates that the source transmits its current update at time $t$, whereas $\mathrm{a}_0^\phi(t)=0$ indicates that the source remains idle. This policy must satisfy an update rate constraint imposed by the resource limitations at the source. We define here (and use throughout the paper) two classes of stationary policies for the source: 

\begin{itemize}
    \item \emph{Threshold policy}: 
    Under a threshold policy $\phi_{\mathcal{T}}$, the source transmits only when the \textsc{VAoI} at node $1$, $\Delta^{\phi_\mathcal{T}}_1(t)$, reaches or exceeds a threshold $\Delta_{\mathcal{T}}$:
    \begin{align*}
        \mathrm{a}^{\phi_{\mathcal{T}}}_0(t) = \mathbf{1} \left\{\Delta^{\phi_\mathcal{T}}_1(t) \geq \Delta_{\mathcal{T}} \right\},
    \end{align*}
    where $\mathbf{1}\{\cdot\}$ is the indicator function. This policy requires \textsc{ACK} feedback from node $1$, which enables the source to calculate $\Delta^{\phi_\mathcal{T}}_1(t)$ and compare it with $\Delta_{\mathcal{T}}$.

    \item \emph{Randomized policy}: 
    Under a randomized policy $\phi_{\mathcal{R}}$, the source transmits independently in each time slot with probability $\uppsi$, i.e., $\mathbb{P}\!\left(\mathrm{a}^{\phi_{\mathcal{R}}}_0(t) \!=\! 1\right) \!=\! \uppsi$.
    This policy is version-agnostic and does not require feedback from node~$1$. 
\end{itemize}

\subsubsection{Forwarding Policy at the Intermediate Nodes}
Intermediate nodes, which are not rate-constrained, adapt their forwarding policies to maximize update delivery based on \textsc{ACK} feedback. If feedback is absent, intermediate nodes transmit continuously in each time slot, as they are unaware of previous successes. We call this policy \emph{always-update} and denote it by $\bm{\theta}^{\textsc{au}}=\theta_{1:N}^{\textsc{au}}=\{ \theta_1^{\textsc{au}},\theta_2^{\textsc{au}},\dots,\theta_N^{\textsc{au}} \}$. If feedback is available, we define two policies: \emph{\textsc{VAoI}-aware} and \emph{feedback-aware}, denoted by $\bm{\theta}^{\textsc{va}}=\theta_{1:N}^{\textsc{va}}$ and $\bm{\theta}^{\textsc{fa}}=\theta_{1:N}^{\textsc{fa}}$, respectively. According to the \textsc{VAoI}-aware policy, an intermediate node keeps transmitting until its \textsc{VAoI} is lower than that of the subsequent node, i.e., until its latest version has been delivered to the next node: 
\begin{align*}
    \mathrm{a}^{\phi,\bm{\theta}^{\textsc{va}}}_i \!(t) \!=\! \mathbf{1} \left\{\Delta^{\phi,\bm{\theta}^{\textsc{va}}}_i \!(t) \!<\! \Delta^{\phi,\bm{\theta}^{\textsc{va}}}_{i+1} \!(t) \right \}, \quad i \!\in\! \{1,2,\dots,N\}.
\end{align*}

Under the feedback-aware policy, a node remains idle only when its stored update has been successfully delivered and no newer update has arrived, ensuring persistent retransmissions until success, irrespective of \textsc{VAoI}. Formally, for node $i \in \{1,2,\dots,N\}$ at time $t+1$:
\begin{align*}
    \mathrm{a}^{\phi, \bm{\theta}^{\textsc{fa}}}_i \!(t+1) \!=\! \left\{ \mathrm{a}^{\phi, \bm{\theta}^{\textsc{fa}}}_{i-1} \!(t) h_{i-1}(t) \right\} \vee \left\{ \mathrm{a}^{\phi, \bm{\theta}^{\textsc{fa}}}_i \!(t) \big(1\!-\!h_i(t)\big) \right\},
\end{align*}
where $\vee$ denotes the logical \textsc{OR} operator. 

\textit{Remark:} The feedback-aware policy is version-agnostic and differs from the \textsc{VAoI}-aware forwarding policy. Specifically, under a feedback-aware policy, the same version may be received multiple times and is forwarded upon each reception. In contrast, a \textsc{VAoI}-aware policy avoids retransmissions by preventing a previously delivered version from being transmitted again. However, when each received update from the source corresponds to a distinct new version, for example, under the threshold policy $\phi = \phi_{\mathcal{T}}$ with $\Delta_{\mathcal{T}} > 0$ at the source, as considered in this paper, the two policies become identical. Consequently, the feedback-aware and \textsc{VAoI}-aware forwarding policies coincide when the source operates under a threshold policy, and for notational convenience, we may therefore use $\bm{\theta}^{\textsc{fa}}$ to also represent $\bm{\theta}^{\textsc{va}}$ in this case.

These policies naturally cover four scenarios with respect to version awareness and feedback availability at intermediate nodes, as summarized in Table~\ref{tab_Policies}.

\begin{table}[tb]
\caption{Update policies and Feedback Requirements}
\label{tab_Policies}
\centering
\setlength{\tabcolsep}{3pt} 
\renewcommand{\cellset}{\renewcommand{\arraystretch}{0.9}} 

\begin{tabular}{@{}llccc@{}}
\toprule
\textbf{\makecell[l]{Scenarios}} & 
  \textbf{\makecell[l]{Policies\\ $(\phi,\bm{\theta})$}} & \textbf{\makecell[l]{Variables \\ $\left(\bar{\eta}^{\phi,\bm{\theta}}_i,\bar{\Delta}^{\phi,\bm{\theta}}_i\right)$}} &
  \textbf{\makecell[c]{\textsc{ACK} from\\ node $1$}} & 
  \textbf{\makecell[c]{\textsc{ACK} from\\ node $i \geq 2$}} \\ \midrule
\makecell[l]{1: Version-aware,\\ \hphantom{1:~}no feedback} & 
  $(\phi_\mathcal{T},\bm{\theta}^{\textsc{au}})$ & $\left(1,\bar{\Delta}^{\phi_{\mathcal{T}}}_i\right)$ &
  \cmark & 
  \xmark \\ 
\makecell[l]{2: Version-aware,\\ \hphantom{2:~}feedback enabled} & 
  $(\phi_\mathcal{T},\bm{\theta}^{\textsc{fa}})$ & $\left(\bar{\eta}^{\phi_{\mathcal{T}}}_i,\bar{\Delta}^{\phi_{\mathcal{T}}}_i\right)$ &
  \cmark & 
  \cmark \\ 
\makecell[l]{3: Version-agnostic,\\ \hphantom{3:~}no feedback} & 
  $(\phi_\mathcal{R},\bm{\theta}^{\textsc{au}})$ & $\left(1,\bar{\Delta}^{\phi_{\mathcal{R}}}_i\right)$ &
  \xmark & 
  \xmark \\ 
\makecell[l]{4: Version-agnostic,\\ \hphantom{4:~}feedback enabled} & 
  $(\phi_\mathcal{R},\bm{\theta}^{\textsc{fa}})$ & $\left(\bar{\eta}^{\phi_{\mathcal{R}}}_i,\bar{\Delta}^{\phi_{\mathcal{R}}}_i\right)$ &
  \xmark & 
  \xmark \\ \bottomrule
\end{tabular}
\end{table}

\subsection{Metrics: Update Rate and VAoI}

Let us define the cumulative number of transmissions by node $i \in \{0,1,2,\dots,N\}$ up to time $t$ as $\eta^{\phi,\bm{\theta}}_i(t) \in \mathbb{N}_0$, given by:
\begin{align}
    \eta^{\phi,\bm{\theta}}_i(t+1) = \eta^{\phi,\bm{\theta}}_i(t) + \mathrm{a}^{\phi,\bm{\theta}}_i(t),
\end{align}
where $\eta^{\phi,\bm{\theta}}_i(0)=0$, $\mathbb{N}_0$ denotes the set of non-negative integers, i.e., $\mathbb{N}_0 = \{0, 1, 2, \dots\}$, and $\phi$ and $\bm{\theta}$ denote the update and forwarding policies adopted at the source and intermediate nodes, respectively. Then, the average number of transmissions, or \emph{update rate}, at node $i$ is defined as follows:
\begin{gather}
    \bar{\eta}^{\phi,\bm{\theta}}_i =  \lim_{T \to \infty} \frac{1}{T} \mathbb{E} \!\left[\eta^{\phi,\bm{\theta}}_i(T)\right] = \lim_{T \to \infty} \frac{1}{T} \mathbb{E} \!\left[\sum_{t=0}^{T-1} \mathrm{a}^{\phi,\bm{\theta}}_i(t) \right]\!. \label{eqn_AvgRateNode}
\end{gather}

The \textsc{VAoI} at node $i$ under policy $(\phi,\bm{\theta})$, denoted by $\Delta^{\phi,\bm{\theta}}_i(t)$, is defined as the number of versions by which node $i$ lags behind the source node $0$:
\begin{align}
    \Delta^{\phi,\bm{\theta}}_i(t) = V_0(t) - V^{\phi,\bm{\theta}}_i(t),
\end{align}
where $V_0(t)$ denotes the current version at the source node $0$, and $V^{\phi,\bm{\theta}}_i(t)$ denotes the version stored at node $i$. The average \textsc{VAoI} at node $i$ under policy $(\phi,\bm{\theta})$, is defined as follows:
\begin{gather}
    \bar{\Delta}^{\phi,\bm{\theta}}_i = \lim_{T \to \infty} \frac{1}{T} \mathbb{E} \!\left[\sum_{t=0}^{T-1} \Delta^{\phi,\bm{\theta}}_i(t) \right]\!.
    \label{eqn_AvgVAoINode}
\end{gather}

Fig.~\ref{fig_VAoI_Evolution} illustrates a sample evolution of \textsc{VAoI} in a two-hop network. The source generates new versions at time slots $0$, $2$, $3$, $6$, and $8$, and its version index $V_0(t)$ increases by one at the beginning of each subsequent time slot. The versions at node $1$, denoted by $V^{\phi}_1(t)$, are updated according to the update policy and the transmissions from node $0$, i.e., $\phi$. If a transmission at time $t$ is successful, then $V^{\phi}_1(t+1) \!=\! V_0(t)$; otherwise, $V^{\phi}_1(t+1) \!=\! V^{\phi}_1(t)$. The corresponding \textsc{VAoI}, $\Delta^{\phi}_1(t) \!=\! V_0(t) \!-\! V^{\phi}_1(t)$, is listed in parentheses in the second row of the table and plotted in blue. Similarly, node $2$ stores the versions $V^{\phi,\bm{\theta}}_2(t)$ received from node $1$, which are transmitted in every slot in this example. In this case, $V^{\phi,\bm{\theta}}_2(t+1) \!=\! V^{\phi}_1(t)$ if the transmission at time $t$ is successful; otherwise, $V^{\phi,\bm{\theta}}_2(t+1) \!=\! V^{\phi,\bm{\theta}}_2(t)$. The \textsc{VAoI} at node $2$, $\Delta^{\phi,\bm{\theta}}_2(t) \!=\! V_0(t) \!-\! V^{\phi,\bm{\theta}}_2(t)$, is shown in parentheses in the third row of the table and plotted in green in Fig. \ref{fig_VAoI_Evolution}.

	\begin{figure}[!t]
		\centering
		\includegraphics[scale=0.46]{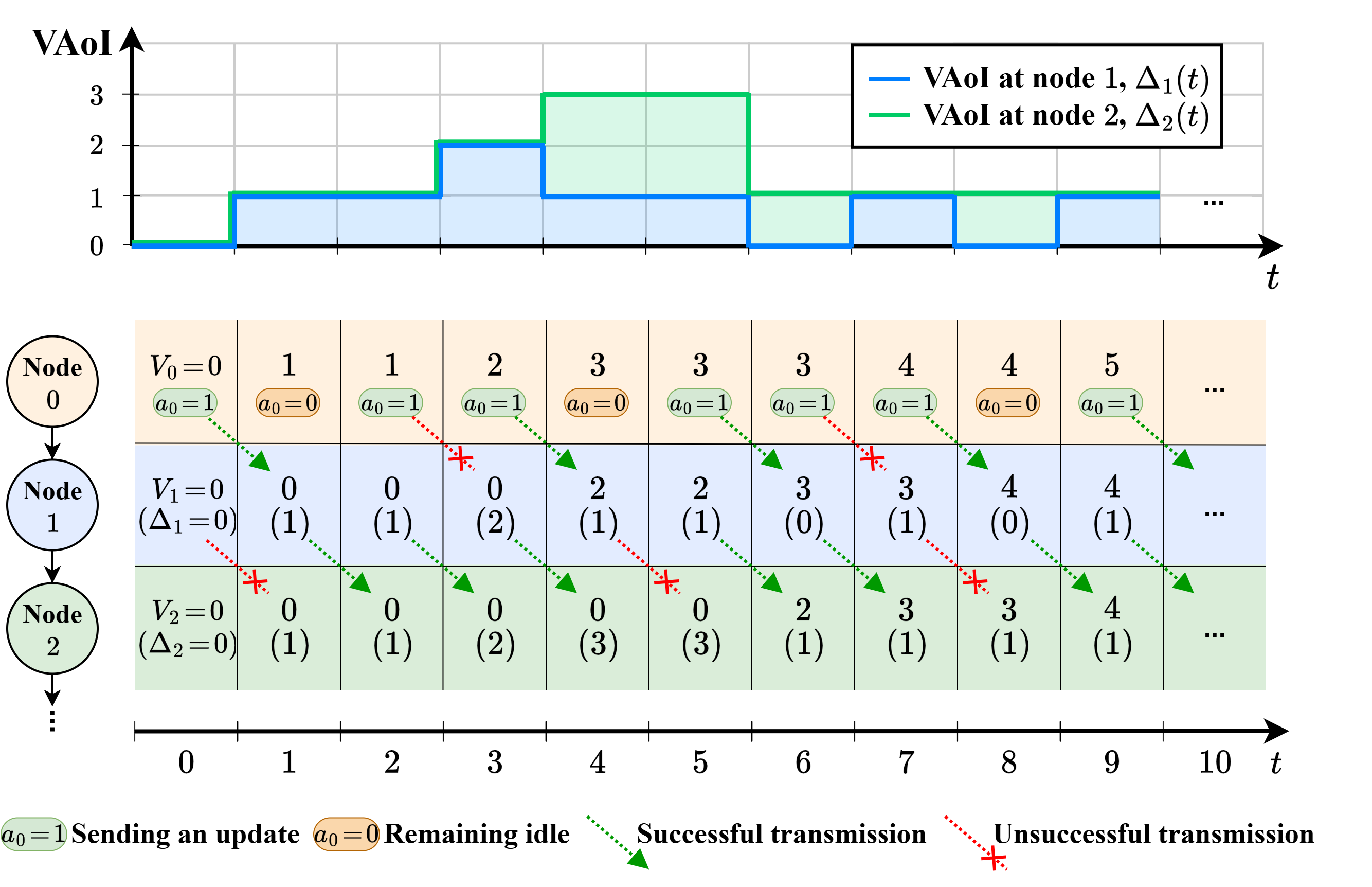} 
		\caption{Evolution of \textsc{VAoI} within the network over time.}
		\label{fig_VAoI_Evolution}
	\end{figure}

\subsection{Problem Formulation}
\label{sec_ProblemFormulation}
Our goal is to minimize \textsc{VAoI} at the destination node while satisfying the rate constraint at the source and minimizing the communication overhead at the intermediate nodes. This problem is formulated as a bi-level optimization problem, where the upper-level problem determines the optimal source policy $\phi^*$, and the lower-level problem determines the optimal update rates at the intermediate nodes under the optimal forwarding policies $\bm{\theta^*}=\theta_{1:N}^*=\{\theta_1^*, \theta_2^*,\dots,\theta_N^*\}$.  The bi-level optimization problem is formally defined as follows:
\medskip

\textbf{$\mathcal{P}_1$: \textsc{VAoI} Optimization}
\begin{align}
\min_{\phi \in \Phi} \quad & \bar{\Delta}_{N+1}^{\phi,\bm{\theta}^{\textsc{au}}} \\
\text{subject to} \quad & \bar{\eta}_0^\phi \le \psi    \label{eq_constraint}
\end{align}
where $\bm{\theta}^{\textsc{au}}=\theta_{1:N}^{\textsc{au}}$ is the always-update forwarding policy, and the optimal policy at the source and the optimal \textsc{VAoI} at the destination are denoted by $\phi^*$ and $\bar{\Delta}_{N+1}^{\phi^*\!,\bm{\theta}^{\textsc{au}}}$, respectively.
\medskip

\textbf{$\mathcal{P}_2$: Update Rate Optimization}
\begin{align}
\min_{\bm{\theta}} \quad & \left\{ \bar{\eta}_1^{\phi^*\!,\bm{\theta}}, \bar{\eta}_2^{\phi^*\!,\bm{\theta}}, \dots, \bar{\eta}_N^{\phi^*\!,\bm{\theta}} \right\} \\
\text{subject to} \quad & \bar{\Delta}_{N+1}^{\phi^*\!,\bm{\theta}} = \bar{\Delta}_{N+1}^{\phi^*\!,\bm{\theta}^{\textsc{au}}}
\end{align}
where $\bm{\theta}=\theta_{1:N}=\{ \theta_1, \theta_2,\dots,\theta_N \}$. This is a multi-objective optimization problem, for which the optimal solution, in the absence of feedback mechanisms, is the \emph{always-update} policy. When \textsc{ACK} feedback is available at the intermediate links, the optimal forwarding policy is a \emph{\textsc{VAoI}-aware} policy, whereby each node $i$ transmits until the successful delivery of its currently stored version, i.e., $\bm{\theta}^*=\bm{\theta}^{\textsc{va}}.$

\begin{table}[tbp]
\centering
\caption{Notation and Variable Descriptions}
\label{tab_notation}
\begin{tabular}{cl}
\toprule
\textbf{Variable} & \textbf{Description} \\ \midrule
$t$ & Discrete time index \\
$i, n$ & Network node indices \\
$N$ & Number of intermediate (relaying) nodes \\
$w(t), p_g$ & Version generation process and its probability \\
$V_i(t)$ & Version index at node $i$ at time $t$ \\
$h_i(t)$ & Success indicator of link $i \rightarrow i+1$ \\
$p_i, \epsilon_i$ & Success and erasure probabilities of link $i \rightarrow i+1$ \\
$\upsigma_{i}(t)$ & Successful arrival process at node $i$ \\ \midrule
$\phi$ & Update policy at source node $0$ \\
$\theta_i$ & Update policy at node $i$ \\
$\phi_\mathcal{T}, \Delta_\mathcal{T}$ & Threshold policy and its \textsc{VAoI} threshold \\
$\phi_\mathcal{R}, \psi$ & Randomized policy and its transmission probability \\ \midrule
$\mathrm{a}_i^\phi(t), a_i$ & Update action process and its value at node $i$ \\
$\Delta_i^\phi(t), \delta_i$ & \textsc{VAoI} process and its value at node $i$ \\
$\bar{\Delta}_i^\phi$ & Average \textsc{VAoI} at node $i$ under policy $\phi$ \\
$\eta_i^\phi(t), \bar{\eta}_i^\phi$ & Cumulative updates and update rate at node $i$ \\
$\pi_{\Delta_i}^\phi(\delta_i)$ & Stationary distribution of \textsc{VAoI} at node $i$ \\
$\pi_{\mathrm{a}_i}^\phi(a_i)$ & Stationary distribution of update action at node $i$ \\ \bottomrule
\end{tabular}
\end{table}

\textit{Notation remark:} For the average update rate, we omit the superscript $\bm{\theta}$ in $\bar{\eta}^{\phi,\bm{\theta}}_i$. This is without loss of generality since, under the always-update policy at intermediate nodes, $\bar{\eta}^{\phi,\bm{\theta}^{\textsc{au}}}_i = 1$. Hence, when analyzing update rates at intermediate nodes, we restrict attention to the policy $\bm{\theta}^{\textsc{fa}}$ and drop the superscript for notational simplicity.
Similarly, we omit the superscript $\bm{\theta}$ in $\bar{\Delta}^{\phi,\bm{\theta}}_i$ since, under all considered policies $\bm{\theta}^{\textsc{au}}$, $\bm{\theta}^{\textsc{va}}$, and $\bm{\theta}^{\textsc{fa}}$, intermediate nodes forward all received versions. While the number of forwarded updates may differ across policies, the resulting \textsc{VAoI} is unchanged. The adopted notation conventions are summarized in Table~\ref{tab_Policies}.
For further simplicity, we also omit the superscript $\phi$ from $\mathrm{a}^{\phi}_i(t)$ and $\Delta^{\phi}_i(t)$ in sections where the analysis applies to general source update policies, i.e., independently of the specific choice of $\phi$.
Moreover, we define $\bar{p}_g = 1 - p_g$ and $\bar{p}_i = 1 - p_i$, for $i = 0,1,2,\dots,N$. A comprehensive list of notations is provided in Table~\ref{tab_notation}.

\section{VAoI Optimization}

    We analyze the objective function of $\mathcal{P}_1$ and derive an expression for the \textsc{VAoI} at the destination node, $\bar{\Delta}_{N+1}$, as a function of the \textsc{VAoI} at the first intermediate node, $\bar{\Delta}_{1}$. To begin, we show that the \textsc{VAoI} at each node can be expressed in terms of the \textsc{VAoI} at the preceding node.
	
	\begin{proposition}
		\label{Prop_VAoInodei}
		The \textsc{VAoI} at node $i\!+\!1$ is given by:
		\begin{align}
			\label{eqn_VAoInodei}
			\Delta_{i+1}(t) &= \Delta_{i}(t-m_i) + \xi_{m_i}, \quad i=1,2,\dots,N,
		\end{align}
		where $m_i$ is a Geometric Random Variable (RV) with parameter $p_i$, and  $\xi_{m_i} \mid m_i$ is a Binomial RV with parameters $m_i$ and $p_g$:
		\begin{align}
			\mathbb{P}\!\left(m_i\!=\!\mathscr{\ell}\right)&\!=\!(1\!-\!p_i)^{\mathscr{\ell}-1}p_i, \quad \hfill \mathscr{\ell}=1,2,\dots. \label{eq_GeometricPMF} \\
            \mathbb{P}\!\left(\xi_{m_i}\!=\!r \!\mid\! m_i=k \right)&\!=\!\mathcal{B}(r; k, p_g)\!=\!\binom{k}{r}p_g^{r}(1-p_g)^{k-r}, \notag \\
            &r=0,1,\dots,k. \label{eq_BinomialPMF} 
		\end{align}
	\end{proposition}
	
	\begin{proof}
		The proof is provided in Appendix \ref{Appen_Proof_VAoInodei}.
	\end{proof}
	
	\begin{lemma}
		\label{Lemma_VAoIDestNode}
		The \textsc{VAoI} at the destination node is given by:
		\begin{align}
			\label{eq_VAoInodeNp1}
			\Delta_{N+1}(t)=\Delta_{1} (t-\tau_{N})+\beta_{N},
		\end{align}
		where $\tau_{N} = \sum_{i=1}^{N} \!m_{i}$ and $\beta_{N}=\sum_{i=1}^{N} \!\xi_{m_i}$ are two RVs with expected values 
		$\mathbb{E} \left[\tau_N \right] = \sum_{i=1}^{N} \frac{1}{p_i}$ and $\mathbb{E} \left[\beta_N \right] = p_g \sum_{i=1}^{N} \frac{1}{p_i}$.
	\end{lemma}
	\begin{proof}
		The proof is provided in Appendix \ref{Appen_Proof_VAoIDestNode}.
	\end{proof}
	\vspace{-2pt}
	The variable $\tau_N$ is the \emph{relaying delay} of each version from node $1$ to node $N\!+\!1$ through $N$ relaying nodes, while $\beta_N$ represents the number of version generations at the source during this delay. 
	
	\begin{theorem}
		\label{Theorem_AvgVAoIlastNode}
		The average \textsc{VAoI} of the destination node which is $N+1$ hops away from the source is given by:
		\begin{align}
			\label{eqn_AvgVAoIlastNode}
			\bar{\Delta}_{N+1} = \bar{\Delta}_1 + p_g \sum_{i=1}^{N} \frac{1}{p_i}. 
		\end{align}
	\end{theorem}
	
	\begin{proof}
		The proof is provided in Appendix \ref{Appen_Proof_AvgVAoIlastNode}.
	\end{proof}

    Equation \eqref{eqn_AvgVAoIlastNode} shows that each additional node contributes an additive expected \textsc{VAoI} penalty proportional to its expected relaying delay.
    
    According to Theorem \ref{Theorem_AvgVAoIlastNode}, the \textsc{VAoI} optimization problem $\mathcal{P}_1$ for the destination node is transformed into a \textsc{VAoI} optimization problem at the first intermediate node, as follows:
    \medskip

    \textbf{$\mathcal{P}^\prime_1$: \textsc{VAoI} Optimization for Node 1}
        \begin{align}
            \min_{\phi \in \Phi} \quad & \bar{\Delta}_{1}^{\phi} \\
            \text{subject to} \quad & \bar{\eta}_0^\phi \le \psi
    \end{align}

    Next, we demonstrate that the optimal on-off policy at the source for optimizing the \textsc{VAoI} at node $1$, under an update rate constraint, is a threshold policy.

    \begin{lemma}
		\label{Lemma_OptimalCMDP}
		The optimal on-off scheduling policy for problem $\mathcal{P}^\prime_1$ is a threshold policy.
	\end{lemma}
    
    \begin{proof}
		The proof is provided in Appendix \ref{Sec_CMDP}.
	\end{proof}

    Therefore, the optimal update policy at the source for solving the \textsc{VAoI} optimization problem $\mathcal{P}_1$ is a threshold policy, i.e., $\phi^*=\phi^*_\mathcal{T}$. In the next subsection, we derive the closed-form expressions for $\bar{\Delta}_{1}^{\phi_\mathcal{T}}$ and $\bar{\eta}_0^{\phi_\mathcal{T}}$, and solve the problem $\mathcal{P}^\prime_1$ to obtain the optimal threshold $\Delta^*_{\mathcal{T}}$ in closed form. We also derive the complete stationary distribution of the \textsc{VAoI} and the update rates for node $1$ and the other network nodes under the version-aware communication policies $(\phi_\mathcal{T},\bm{\theta}^{\textsc{va}})$.

\subsection{Analysis of $\Delta^{\phi_{\mathcal{T}}}_1$ and $\bar{\eta}^{\phi_{\mathcal{T}}}_0$}

Under a threshold policy $\phi_{\mathcal{T}}$, node $0$ controls the transmission actions $\mathrm{a}^{\phi_{\mathcal{T}}}_0(t) \in \{0,1\}$ based on the \textsc{VAoI} at node $1$ relative to a threshold $\Delta_\mathcal{T}$:
$\mathrm{a}^{\phi_{\mathcal{T}}}_0(t) = \mathbf{1}\{\Delta^{\phi_{\mathcal{T}}}_1(t) \geq \Delta_\mathcal{T}\}.$
Under this policy, the evolution of $\Delta^{\phi_{\mathcal{T}}}_1(t)$ is independent of the actions of other nodes and is fully characterized by $\mathrm{a}^{\phi_{\mathcal{T}}}_0(t)$ and the system processes, including $w(t)$ and $h_0(t)$. Hence, $\Delta^{\phi_{\mathcal{T}}}_1(t)$ can be modeled as a \textsc{DTMC}, and its marginal stationary distribution, denoted by $\pi^{\phi_{\mathcal{T}}}_{\Delta_1}(\delta_1)$, gives the long-run probability of being in state $\delta_1 \in \mathbb{N}_0$. It is obtained from the balance equations:
\begin{align}
    \pi^{\phi_{\mathcal{T}}}_{\Delta_1}(\delta_1) = \sum_{\delta_1^\prime \in \mathbb{N}_0} \mathbb{P}(\delta_1 | \delta_1^\prime)\,\pi^{\phi_{\mathcal{T}}}_{\Delta_1}(\delta_1^\prime),
\end{align}
with $\sum_{\delta_1 \in \mathbb{N}_0} \pi^{\phi_{\mathcal{T}}}_{\Delta_1}(\delta_1) = 1$, where $\mathbb{P}(\delta_1 | \delta_1^\prime)$ is the transition probability from $\Delta^{\phi_{\mathcal{T}}}_1(t)=\delta_1^\prime$ to $\Delta^{\phi_{\mathcal{T}}}_1(t+1)=\delta_1$. The average \textsc{VAoI} at node $1$ and the update rate at node $0$ then follow from this distribution as:
\begin{align}
    \bar{\Delta}^{\phi_{\mathcal{T}}}_1 &\!=\! \mathbb{E}\left[\Delta^{\phi_{\mathcal{T}}}_1\!(t)\right]\!=\!\sum_{\delta_1 \in \mathbb{N}_0} \delta_1 \pi^{\phi_{\mathcal{T}}}_{\Delta_1}(\delta_1), \\
    \bar{\eta}^{\phi_{\mathcal{T}}}_0 &\!=\! \mathbb{E}\left[\mathrm{a}^{\phi_{\mathcal{T}}}_0\!(t)\right]\!=\!\mathbb{P}(\Delta^{\phi_{\mathcal{T}}}_1\!(t) \geq \Delta_\mathcal{T})\!=\!\!\!\sum_{\delta_1 \!\geq\! \Delta_\mathcal{T}} \!\pi^{\phi_{\mathcal{T}}}_{\Delta_1}(\delta_1).
\end{align}

The \textsc{VAoI} at node $1$, $\Delta^{\phi_{\mathcal{T}}}_1(t)$, measures the difference between the version index available at node $1$ and the version at the source. Its dynamics are governed by version generation at node $0$ and transmissions over the erasure channel. In each slot $t$, the source generates a new version with probability $p_g$, following the Bernoulli process $w(t) \in \{0, 1\}$. Node $0$ transmits its current version, which is successfully received with probability $p_0$.

The \textsc{VAoI} evolution depends on these two processes. If node $1$ successfully receives an update at time $t$, its buffer is refreshed to the source’s current version and the \textsc{VAoI} resets as $\Delta^{\phi_{\mathcal{T}}}_1(t+1) = w(t)$. Otherwise, the \textsc{VAoI} is carried over and may increase due to new version generation, i.e.,
$\Delta^{\phi_{\mathcal{T}}}_1(t+1) = \Delta^{\phi_{\mathcal{T}}}_1(t) + w(t)$. Thus, the \textsc{VAoI} remains unchanged if no new version is generated ($w(t)=0$), and increases by one if a new version is generated ($w(t)=1$), provided that node $1$ does not receive a successful update.

\begin{proposition}
		\label{Prop_StateProbThr}
		The steady-state probability of the \textsc{VAoI} at node $1$ under a threshold policy with threshold $\Delta_\mathcal{T}$ is given by the following:
		\begin{itemize}[leftmargin=1em]
			\item For $\Delta_\mathcal{T} \in \{0,1\}$: 
            \begin{align}
			\pi^{\phi_{\mathcal{T}}}_{\Delta_1}(\delta_1) \!=\!\!
			     \begin{cases}
				\frac{\bar{p}_g p_0}{\alpha_\mathcal{T}}, & \delta_1=0,\\
				 \frac{p_g p_0}{\alpha_\mathcal{T}^2} \left [ \frac{p_g \bar{p}_0}{\alpha_\mathcal{T}} \right ]^{\delta_1-1}, & \delta_1 \geq 1.
			     \end{cases}
		      \end{align}
			\item For $\Delta_\mathcal{T} \geq 2$:
			\begin{align}
				\pi^{\phi_{\mathcal{T}}}_{\Delta_1}(\delta_1) \!=\!\!
				\begin{cases}
					\frac{\bar{p}_g p_0}{(\Delta_\mathcal{T}\!-\!1)p_0+\alpha_\mathcal{T}}, & \delta_1\!=\!0,\\
					\frac{p_0}{(\Delta_\mathcal{T}\!-\!1)p_0+\alpha_\mathcal{T}}, & 1 \!\leq\! \delta_1 \!<\! \Delta_\mathcal{T}, \\
					\frac{p_g}{\alpha_\mathcal{T}} \pi^{\phi_{\mathcal{T}}}_{\Delta_1} \!\! \left(\Delta_\mathcal{T}\!-\!1\right) \! \left (\!\frac{p_g \bar{p}_0}{\alpha_\mathcal{T}} \!\right)^{\delta_1\!-\!\Delta_\mathcal{T}}\!\!\!\!, & \delta_1 \!\geq\! \Delta_\mathcal{T},
				\end{cases}
			\end{align}
		\end{itemize}
		where $\alpha_\mathcal{T}=1- \bar{p}_g \bar{p}_0$.
	\end{proposition}

    \begin{proof}
        The proof is provided in Appendix \ref{Appen_Proof_ThrPolicy_SSProbs}.
    \end{proof}

\begin{lemma}
    \label{Lemma_AvgD1Eta0}
    The average \textsc{VAoI} at node $1$ and the update rate at node $0$ under a threshold policy with threshold $\Delta_\mathcal{T} \geq 1$ are given by:
    \begin{align}
        \label{eqn_AvgD1Threshold}
        \bar{\Delta}^{\phi_{\mathcal{T}}}_1 &= \frac{1}{2} \frac{(\Delta_\mathcal{T}\!-\!1)\Delta_\mathcal{T} p_0}{(\Delta_\mathcal{T}\!-\!1)p_0 + \alpha_\mathcal{T}} + \frac{p_g}{p_0}, \\
        \bar{\eta}^{\phi_{\mathcal{T}}}_0 &= \frac{p_g}{(\Delta_\mathcal{T}-1)p_0+\alpha_\mathcal{T}},
    \end{align}
   where $\alpha_\mathcal{T}=1- \bar{p}_g \bar{p}_0$.
\end{lemma}

\begin{proof}
    The proof is provided in Appendix \ref{Appen_Proof_AvdD1Eta1Thr}.
\end{proof}

It is evident that $\Delta_\mathcal{T}=0$ results in an \emph{always-update} policy at node $0$, where $\bar{\Delta}^{\phi_{\mathcal{T}}}_1=\frac{p_g}{p_0}$ and $\bar{\eta}^{\phi_{\mathcal{T}}}_0=1$.
Using Lemma \ref{Lemma_AvgD1Eta0}, the optimization problem $\mathcal{P}^\prime_1$ can be reformulated as follows:
\medskip 

    \textbf{$\mathcal{P}^{''}_1$: \textsc{VAoI} Optimization for Node 1}
        \begin{align}
            \min_{\Delta_\mathcal{T} \in \mathbb{N}} \quad & \frac{1}{2} \frac{(\Delta_\mathcal{T}\!-\!1)\Delta_\mathcal{T} p_0}{(\Delta_\mathcal{T}\!-\!1)p_0 + \alpha_\mathcal{T}} + \frac{p_g}{p_0} \\
            \text{subject to} \quad & \frac{p_g}{(\Delta_\mathcal{T}-1)p_0+\alpha_\mathcal{T}} \le \psi
    \end{align}
    where its solution is obtained in the following theorem.

\begin{theorem}
		\label{Theorem_OptimalThreshold}
		The optimal threshold policy that minimizes the average \textsc{VAoI} at the network under the rate constraint at the source is a randomized mixture of two threshold policies with thresholds $\Delta_{\mathcal{T}}^\ast$ and $\Delta_{\mathcal{T}}^\ast - 1$, applied with probabilities $\kappa$ and $1 - \kappa$, respectively. The optimal threshold is:
		\begin{align}
			\label{eqn_OptimalThreshold}
			\Delta_{\mathcal{T}}^\ast= \Bigg\lceil \frac{p_g}{p_0}\left( \frac{1}{\psi} - 1 +p_0\right) \!\! \Bigg\rceil,
		\end{align}
		and the corresponding mixing probability $\kappa$ is:
		\begin{align}
			\label{Optimal_gamma}
			\kappa = \frac{f(\Delta_{\mathcal{T}}^\ast \!-\! 1) - \psi}{f(\Delta_{\mathcal{T}}^\ast \!-\! 1) - f(\Delta_{\mathcal{T}}^\ast)},
		\end{align}
		where $f(\Delta_{\mathcal{T}}) = \frac{p_g}{(\Delta_{\mathcal{T}}-1)p_0+\alpha_{\mathcal{T}}}$ for $\Delta_{\mathcal{T}} \geq 1$, and $f(0) = 1$.
	\end{theorem}
	
	\begin{proof}
		The proof is provided in Appendix \ref{Appen_Proof_OptimalThr}.
	\end{proof}
	
	The resulting optimal average \textsc{VAoI} at node $1$ under the \emph{mixed threshold policy} is given by:
    \begin{align}
        \label{eqn_OptimalVAoI_Mix}
        \bar{\Delta}^{\phi^\ast}_1 = \kappa \left[\bar{\Delta}^{\phi_{\mathcal{T}}}_1 \right]_{\Delta_{\mathcal{T}}=\Delta_{\mathcal{T}}^\ast} + \left(1-\kappa\right) \left[\bar{\Delta}^{\phi_{\mathcal{T}}}_1 \right]_{\Delta_{\mathcal{T}}=\Delta_{\mathcal{T}}^\ast-1},
    \end{align}
    where $\left[\bar{\Delta}^{\phi_{\mathcal{T}}}_1 \right]_{\Delta_{\mathcal{T}}=\Delta_{\mathcal{T}}^\ast}$ and $\left[\bar{\Delta}^{\phi_{\mathcal{T}}}_1 \right]_{\Delta_{\mathcal{T}}=\Delta_{\mathcal{T}}^\ast-1}$ are obtained from \eqref{eqn_AvgD1Threshold} by setting $\Delta_{\mathcal{T}}$ to $\Delta_{\mathcal{T}}^\ast$ and $\Delta_{\mathcal{T}}^\ast-1$, respectively. The resulting optimal \textsc{VAoI} at the destination node $N+1$ is then given by \eqref{eqn_AvgVAoIlastNode} as follows:
    \begin{align}
        \bar{\Delta}^{\phi^\ast}_{N+1} = \bar{\Delta}^{\phi^\ast}_1 + p_g \sum_{i=1}^{N} \frac{1}{p_i}.
    \end{align}

    This is the optimal solution to problem $\mathcal{P}_1$, representing the best achievable \textsc{VAoI} at the destination node, as well as at all network nodes:
    $\bar{\Delta}^{\phi^\ast}_{n+1} = \bar{\Delta}^{\phi^\ast}_1 + p_g \sum_{i=1}^{n} \frac{1}{p_i}$, for $n\in\{1,2,\dots,N\}$.
    We now derive the \textsc{VAoI} at the network nodes; specifically, we obtain the distribution of the \textsc{VAoI} for nodes $2,3,\dots,N+1$. 

    \subsection{Distribution of $\{\Delta^{\phi}_i\}_{i=2}^{N}$}

    \begin{theorem}
        \label{Theorem_RecDistVAoIi}
        The stationary distribution of the \textsc{VAoI} at node $i+1$ is given by:
        \begin{align}
            \label{eqn_Dist_Di}
            \pi^{\phi}_{\Delta_{i+1}}(\delta) &= \sum_{\mathscr{\ell}=1}^{\infty} \mathbb{P}(m_i\!=\!\mathscr{\ell}) \left[ \mathcal{B}(\delta; \mathscr{\ell}, p_g) \circledast \pi^{\phi}_{\Delta_{i}}(\delta) \right] \\
            &= \sum_{\mathscr{\ell}=1}^{\infty} \sum_{r=0}^{\mathscr{\ell}} \mathbb{P}(m_i\!=\!\mathscr{\ell}) \mathbb{P}(\xi_{m_i}\!=\!r \!\mid\! m_i\!=\!\mathscr{\ell}) \pi^{\phi}_{\Delta_{i}}(\delta-r), \notag
        \end{align}
        for $i \in \{1,2,\dots,N\}$ and a source update policy $\phi$, where $\circledast$ denotes the convolution operator.
    \end{theorem}
    \begin{proof}
		The proof is provided in Appendix \ref{Appen_Proof_RecDistVAoIi}.
	\end{proof}

    Given $\pi^{\phi}_{\Delta_{1}}(\delta_1)$ under the adopted policy $\phi$ at the source, the \textsc{VAoI} distribution at the remaining network nodes can be derived recursively.

    \section{\textsc{VAoI}-aware Update Rates Under $(\phi_{\mathcal{T}},\bm{\theta}^{\textsc{va}})$}

    In this section, we derive the update rates at the network nodes for problem $\mathcal{P}_2$, under a threshold policy $\phi_{\mathcal{T}}$ with an arbitrary threshold $\Delta_{\mathcal{T}} \geq 1$ at the source (where $\phi^\ast$ is a special case corresponding to the optimal policy) and a \textsc{VAoI}-aware policy $\bm{\theta}^{\textsc{va}}$ at the intermediate nodes, which results in the best achievable \textsc{VAoI} at the destination.

    \subsection{Analysis of $\bar{\eta}^{\phi_{\mathcal{T}}}_1$}

    \begin{figure*}[t]
        \centering
        \includegraphics[width=0.76\linewidth]{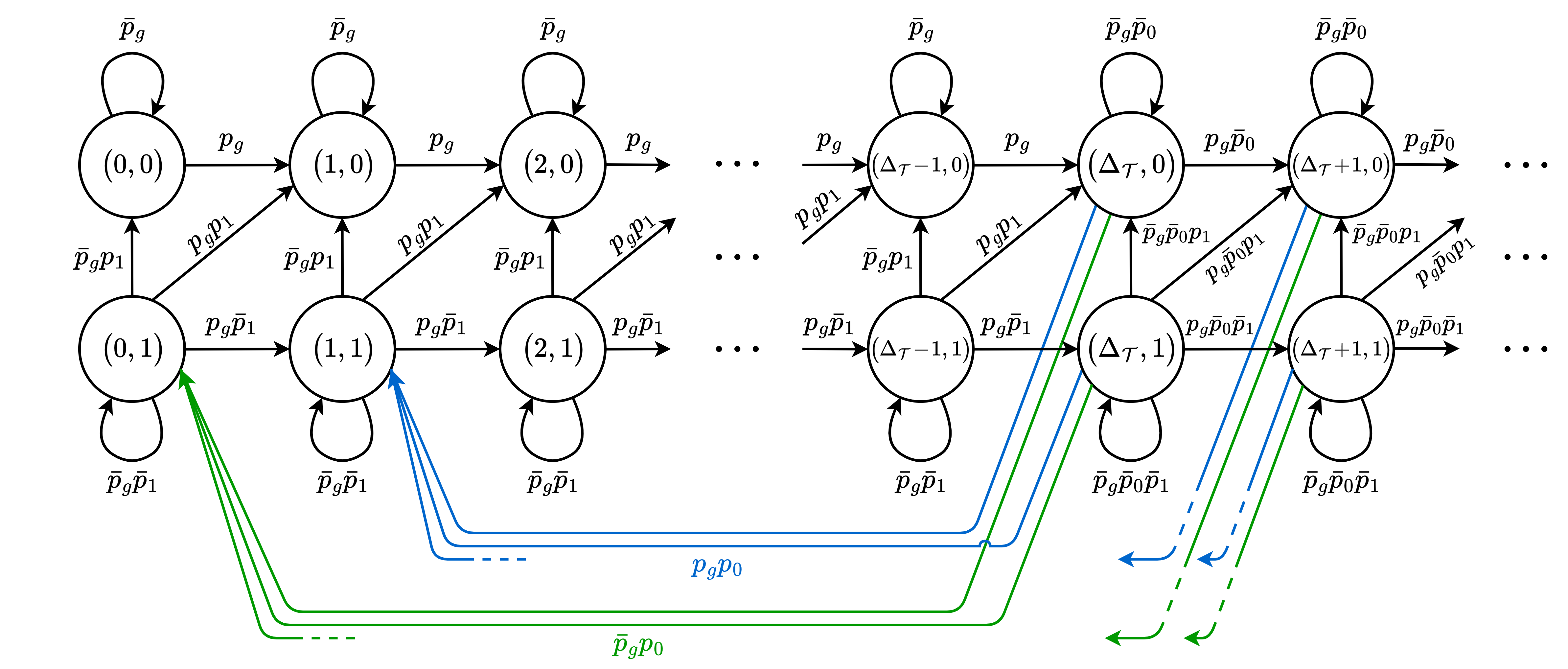}
        \caption{Markov Chain of $(\Delta^{\phi_{\mathcal{T}}}_1,\mathrm{a}^{\phi_{\mathcal{T}}}_1)$ under $(\phi_{\mathcal{T}},\theta^{\textsc{va}})$.}
        \label{fig_D1a1_DTMC}
    \end{figure*}

    \begin{theorem}
        \label{Theorem_AvgEta1_Thr}
        The update rate at node $1$ under $(\phi_{\mathcal{T}},\theta^{\textsc{va}}_1)$ is given by:
        \begin{align}
            \bar{\eta}^{\phi_{\mathcal{T}}}_1 = \frac{p_0}{p_1} \left(1-\frac{p_0 \Theta^{\Delta_\mathcal{T}}}{1-\bar{p}_0 \bar{p}_1}\right) \bar{\eta}^{\phi_{\mathcal{T}}}_0,
        \end{align}
        where $\Theta=\frac{p_g \bar{p}_1}{1-\bar{p}_g \bar{p}_1}$.
    \end{theorem}

\begin{proof}
    The update rate $\bar{\eta}^{\phi_{\mathcal{T}}}_1 = \mathbb{E}[\mathrm{a}^{\phi_{\mathcal{T}}}_1] = \mathbb{P}(\mathrm{a}^{\phi_{\mathcal{T}}}_1=1)$ is obtained by deriving the stationary distribution of the joint variables $(\Delta^{\phi_\mathcal{T}}_1,\mathrm{a}^{\phi_\mathcal{T}}_1)$, denoted by $\pi^{\phi_{\mathcal{T}}}_{(\Delta_1,\mathrm{a}_1)}(\delta_1,a_1)$, which is obtained from a 2-D \textsc{DTMC} in Fig. \ref{fig_D1a1_DTMC}. The details of the proof are provided in Appendix \ref{Appen_Proof_AvgEta1_Thr}.
\end{proof}

\subsection{Analysis of forwarding rates $\{\bar{\eta}^{\phi_{\mathcal{T}}}_i\}_{i=2}^{N}$}

 We propose an iterative algorithm to derive the transition probabilities of the state vector $\mathbf{x}_n = \big(\Delta_1, \mathrm{a}_1, \mathrm{a}_2, \dots, \mathrm{a}_n\big),\ n \in \{1,2,3,\dots,N\}$. The action at each node is independent of the subsequent nodes; therefore, we can factorize the transition probability as:
 \begin{align}
    \label{eqn_IterativeUpdateRate}
    \mathbb{P}\left(\mathbf{x}_{n+1}^\prime \mid \mathbf{x}_{n+1}\right) & =\mathbb{P}\left(\mathbf{x}_{n}^\prime, a_{n+1}^\prime \mid \mathbf{x}_{n},a_{n+1}\right)  \\
     &= \mathbb{P}\left(\mathbf{x}_n^\prime \mid \mathbf{x}_n\right) \mathbb{P}\left(a_{n+1}^\prime \mid \mathbf{x}_n^\prime,\mathbf{x}_n,a_{n+1}\right) \notag \\
     & = \mathbb{P}\left(\mathbf{x}_n^\prime \mid \mathbf{x}_n\right) K_{\mathbf{x}_n \rightarrow \mathbf{x}_n^\prime}\left(a_{n+1}^\prime \mid a_{n+1}\right), \notag
 \end{align}
 where we have defined the kernel function $K_{\mathbf{x}_n \rightarrow \mathbf{x}_n^\prime}\left(a_{n+1}^\prime \mid a_{n+1}\right) = \mathbb{P}\left(a_{n+1}^\prime \mid \mathbf{x}_n^\prime,\mathbf{x}_n,a_{n+1}\right)$ as the transition probability of the action at node $n+1$ from $a_{n+1}$ to $a_{n+1}^\prime$, given that the previous nodes transition from $\mathbf{x}_n$ to $\mathbf{x}_n^\prime$. This kernel can be derived by introducing a variable $\upsigma_{n+1} \in \{0,1\}$ that captures the successful arrival of an update at node $n+1$. We define a parameter $\rho^{(n+1)}_{\mathbf{x}_n \rightarrow \mathbf{x}_n^\prime}$ as the probability that $\upsigma_{n+1}$ equals $1$ when the previous nodes transition from $\mathbf{x}_n$ to $\mathbf{x}_n^\prime$, i.e., $\rho^{(n+1)}_{\mathbf{x}_n \rightarrow \mathbf{x}_n^\prime} = \mathbb{P}(\upsigma_{n+1} = 1 \mid \mathbf{x}_{n}^\prime , \mathbf{x}_{n})$. This probability can be expressed recursively as follows:
 \begin{align}
    \label{eqn_AlphaVariable}
     \rho^{(n+1)}_{\mathbf{x}_n \rightarrow \mathbf{x}_n^\prime} = 
     \begin{cases}
         0 & a_n = 0, a_n^\prime \in \{0, 1\}, \\
         1 & a_n = 1, a_n^\prime = 0, \\
         p_{i-1} \rho^{(n)}_{\mathbf{x}_{n-1} \rightarrow \mathbf{x}_{n-1}^\prime} & a_n = 1, a_n^\prime = 1.
     \end{cases} 
 \end{align}
 with the following initialization, where $\mathbf{x}_0 = \Delta_1$:
 \begin{align}
    \label{eq_InitRho1}
     \rho^{(1)}_{\mathbf{x}_{0} \rightarrow \mathbf{x}_{0}^\prime} =
     \begin{cases}
         1 & \Delta_1 \geq \Delta_\mathcal{T}, \ \Delta^\prime_1 \leq 1, \\
         0 & \text{otherwise}.
     \end{cases}
 \end{align}

 By conditioning on whether there is an arrival to node $n+1$, i.e., $\upsigma_{n+1}$, the kernel function is derived as follows:
 \begin{align}
    \label{eqn_TransitionKernel_Thr}
     K&_{\mathbf{x}_n \rightarrow \mathbf{x}_n^\prime}\left(a_{n+1}^\prime \mid a_{n+1}\right) \\
     &= \left(1-\rho^{(n+1)}_{\mathbf{x}_n \rightarrow \mathbf{x}_n^\prime}\right) \mathbb{P}\left(a_{n+1}^\prime \mid \mathbf{x}_n^\prime,\mathbf{x}_n,a_{n+1},\upsigma_{n+1}=0\right) \notag \\
     &\quad + \rho^{(n+1)}_{\mathbf{x}_n \rightarrow \mathbf{x}_n^\prime} \mathbb{P}\left(a_{n+1}^\prime \mid \mathbf{x}_n^\prime,\mathbf{x}_n,a_{n+1},\upsigma_{n+1}=1\right) \notag \\
     &= 
     \begin{cases}
         1-\rho^{(n+1)}_{\mathbf{x}_n \rightarrow \mathbf{x}_n^\prime} & a_{n+1} = 0, a_{n+1}^\prime = 0, \\
         \rho^{(n+1)}_{\mathbf{x}_n \rightarrow \mathbf{x}_n^\prime} & a_{n+1} = 0, a_{n+1}^\prime = 1, \\
         p_i \left(1-\rho^{(n+1)}_{\mathbf{x}_n \rightarrow \mathbf{x}_n^\prime}\right) & a_{n+1} = 1, a_{n+1}^\prime = 0, \\
         1-p_i \left(1-\rho^{(n+1)}_{\mathbf{x}_n \rightarrow \mathbf{x}_n^\prime}\right) & a_{n+1} = 1, a_{n+1}^\prime = 1.
     \end{cases} \notag
 \end{align}

 Using this kernel function and starting from $\mathbf{x}_1=\left(\Delta_1,\mathrm{a}_1\right)$, whose transition probabilities are known (as depicted in Fig.~\ref{fig_D1a1_DTMC}), equation~\eqref{eqn_IterativeUpdateRate} can be computed iteratively to obtain the transition probabilities of the state vector $\mathbf{x}_n$, which contains all update actions, as summarized in Algorithm \ref{Alg_IterativeTransitionsMultiHop_Thr}.

Given the transition probabilities, by solving the global balance equations 
$\pi^{\phi_{\mathcal{T}}}_{\mathbf{x}_N}(\mathbf{x}^\prime_N) = \sum_{\mathbf{x}_N} \pi^{\phi_{\mathcal{T}}}_{\mathbf{x}_N}(\mathbf{x}_N) \mathbb{P}(\mathbf{x}^\prime_N \mid \mathbf{x}_N)$ 
subject to the normalization constraint 
$\sum_{\mathbf{x}_N} \pi^{\phi_{\mathcal{T}}}_{\mathbf{x}_N}(\mathbf{x}_N) = 1$, the joint stationary distribution 
$\pi^{\phi_{\mathcal{T}}}_{\mathbf{x}_N}(\mathbf{x}_N) = \pi^{\phi_{\mathcal{T}}}_{\mathbf{x}_N}(\delta_1, a_1, a_2, \dots, a_N)$ is obtained. 
From this, the marginal stationary distribution for any specific node $i$ is derived by summing over all other state variables, specifically 
$\pi^{\phi_{\mathcal{T}}}_{\mathrm{a}_i}(a_i) = \sum_{\delta_1} \sum_{a_{j \neq i}} \pi^{\phi_{\mathcal{T}}}_{\mathbf{x}_N}(\delta_1, a_1, \dots, a_N)$. 
Consequently, the update rate at node $i$ is calculated as 
$\bar{\eta}^{\phi_{\mathcal{T}}}_i = \sum_{a_i \in \{0,1\}} a_i \pi^{\phi_{\mathcal{T}}}_{\mathrm{a}_i}(a_i) = \pi^{\phi_{\mathcal{T}}}_{\mathrm{a}_i}(1)$, 
where $\pi^{\phi_{\mathcal{T}}}_{\mathrm{a}_i}(a_i)$ represents the marginal steady-state probability that node $i$ is in state $a_i \in \{0,1\}$. This procedure is presented in the second part of Algorithm \ref{Alg_IterativeTransitionsMultiHop_Thr} for completeness.

\begin{algorithm}[tb]
\caption{Transition Probability Derivation and Stationary Distribution for Multi-hop \textsc{DTMC} under Policies $(\phi_{\mathcal{T}},\bm{\theta}^{\textsc{va}})$}
\begin{algorithmic}[1]
\Require Transition probabilities of $(\Delta_1, \mathrm{a}_1)$ and system parameters $\{N, \Delta_\mathcal{T}, p_g,p_0,p_1,\dots,p_N\}$
\State \textbf{Initialize:} $\mathbf{x}_1 = (\Delta_1, \mathrm{a}_1)$
\State \textbf{Initialize:} $\rho^{(1)}_{\mathbf{x}_0 \rightarrow \mathbf{x}_0^\prime}$ using equation \eqref{eq_InitRho1}
\For{$n = 1$ \textbf{to} $N-1$}
    \State Compute $\rho^{(n+1)}_{\mathbf{x}_n \rightarrow \mathbf{x}_n^\prime}$ using equation \eqref{eqn_AlphaVariable}
    \State Compute kernel $K_{\mathbf{x}_n \rightarrow \mathbf{x}^\prime_n}(a_{n+1}^\prime \mid a_{n+1})$ using equation \eqref{eqn_TransitionKernel_Thr}
    \State Compute the transition probability $\mathbb{P}(\mathbf{x}_{n+1}^\prime \mid \mathbf{x}_{n+1})$ using equation \eqref{eqn_IterativeUpdateRate} where $\mathbf{x}_{n+1} = \left( \mathbf{x}_n , \mathrm{a}_{n+1} \right)$.
\EndFor
\State \Return $\mathbb{P}(\mathbf{x}_{n}^\prime \mid \mathbf{x}_{n})$ for $n=1, 2, \dots, N$.

\algrule
\raggedright \textbf{Stationary Distribution and Update Rates}
\medskip

\setcounter{ALG@line}{0}
\State Solve the global balance equations $\pi^{\phi_{\mathcal{T}}}_{\mathbf{x}_N}(\mathbf{x}^\prime_N) = \sum_{\mathbf{x}_N} \pi^{\phi_{\mathcal{T}}}_{\mathbf{x}_N}(\mathbf{x}_N) \mathbb{P}(\mathbf{x}^\prime_N \mid \mathbf{x}_N)$ subject to $\sum_{\mathbf{x}_N} \pi^{\phi_{\mathcal{T}}}_{\mathbf{x}_N}(\mathbf{x}_N) = 1$ to obtain the joint stationary distribution $\pi^{\phi_{\mathcal{T}}}_{\mathbf{x}_N}(\mathbf{x}_N)$.
\For{each node $i = 1$ \textbf{to} $N$}
    \State Compute the marginal stationary distribution: $\pi^{\phi_{\mathcal{T}}}_{\mathrm{a}_i}(a_i) = \sum_{\delta_1} \sum_{a_{j \neq i}} \pi^{\phi_{\mathcal{T}}}_{\mathbf{x}_N}(\delta_1, a_1, \dots, a_N)$
    \State Evaluate the node update rate: $\bar{\eta}^{\phi_{\mathcal{T}}}_i = \pi^{\phi_{\mathcal{T}}}_{\mathrm{a}_i}(1)$
\EndFor
\State \Return $\bar{\eta}^{\phi_{\mathcal{T}}}_i$ for $i=1, 2, \dots, N$.
\end{algorithmic}
\label{Alg_IterativeTransitionsMultiHop_Thr}
\end{algorithm}

\section{VAoI and Update Rates Under $(\phi_{\mathcal{R}},\bm{\theta}^{\textsc{fa}})$}
In this section, we analyze the \textsc{VAoI} and update rates under \emph{version-agnostic} policies: the randomized policy $\phi_{\mathcal{R}}$ at the source and the feedback-aware policy $\bm{\theta}^{\textsc{fa}}$ at the intermediate nodes. These results serve as a baseline for comparison with the \textsc{VAoI} and update rates obtained under the \emph{version-aware} policies $(\phi_{\mathcal{T}},\bm{\theta}^{\textsc{va}})$ derived in the previous section.

\subsection{Analysis of $\bar{\Delta}^{\phi_{\mathcal{R}}}_1$ and $\bar{\eta}^{\phi_{\mathcal{R}}}_0$}

The update rate at the source node under the randomized policy is obtained as follows:
\begin{align}
    \bar{\eta}^{\phi_{\mathcal{R}}}_0 = \mathbb{E}\left[ \mathrm{a}^{\phi_{\mathcal{R}}}_0(t)\right] = \mathbb{P}\left(\mathrm{a}^{\phi_{\mathcal{R}}}_0(t)=1\right) = \psi.
\end{align}

    \begin{proposition}
    \label{Prop_StateProbRandom}
    The stationary distribution and average \textsc{VAoI} at node $1$ under a randomized policy with transmission probability $\psi$ are given by:
    \begin{gather}
        \pi^{\phi_{\mathcal{R}}}_{\Delta_1}(\delta_1) =
        \begin{cases}
            \frac{\psi \bar{p}_g p_0}{\alpha_\mathcal{R}}, & \delta_1=0,\\
            \frac{\psi p_g p_0}{\alpha_\mathcal{R}^2} \left[ \frac{p_g (1-\psi p_0)}{\alpha_\mathcal{R}} \right]^{\delta_1-1}, & \delta_1 \geq 1,
        \end{cases} \\
        \label{eqn_AvgVAoI1_RS}
        \bar{\Delta}^{\phi_{\mathcal{R}}}_1 = \frac{p_g}{\psi p_0},
    \end{gather}
    where $\alpha_\mathcal{R}=1-\bar{p}_g(1-\psi p_0)$.
    \end{proposition}
    
	\begin{proof}
		The proof is provided in in Appendix \ref{Appen_Proof_RSpolicy_SSProbs}.
	\end{proof}

    Given $\pi^{\phi_\mathcal{R}}_{\Delta_{1}}(\delta_1)$ and $\bar{\Delta}^{\phi_{\mathcal{R}}}_1$, the \textsc{VAoI} distribution and its average at the remaining nodes, namely $\{\pi^{\phi_\mathcal{R}}_{\Delta_{i}}(\delta_i)\}_{i=2}^{N+1}$ and $\{\bar{\Delta}^{\phi_{\mathcal{R}}}_i\}_{i=2}^{N+1}$, can be derived using \eqref{eqn_Dist_Di} and \eqref{eqn_AvgVAoIlastNode}, respectively.

\subsection{Analysis of $\bar{\eta}^{\phi_{\mathcal{R}}}_1$}

From the \textsc{DTMC} of $\mathrm{a}^{\phi_\mathcal{R}}_1$ depicted in Fig.~\ref{fig_a1_RS_DTMC}, the distribution and the average update rate at node $1$ are obtained as follows:

\begin{align}
    \bar{\eta}^{\phi_{\mathcal{R}}}_1 = \pi^{\phi_{\mathcal{R}}}_{\mathrm{a}_1}(1) = 1 - \pi^{\phi_{\mathcal{R}}}_{\mathrm{a}_1}(0) = \frac{\psi p_0}{1-(1-\psi p_0) \bar{p}_1}.
\end{align}

\begin{figure}[H]
    \centering
    \includegraphics[width=0.55\linewidth]{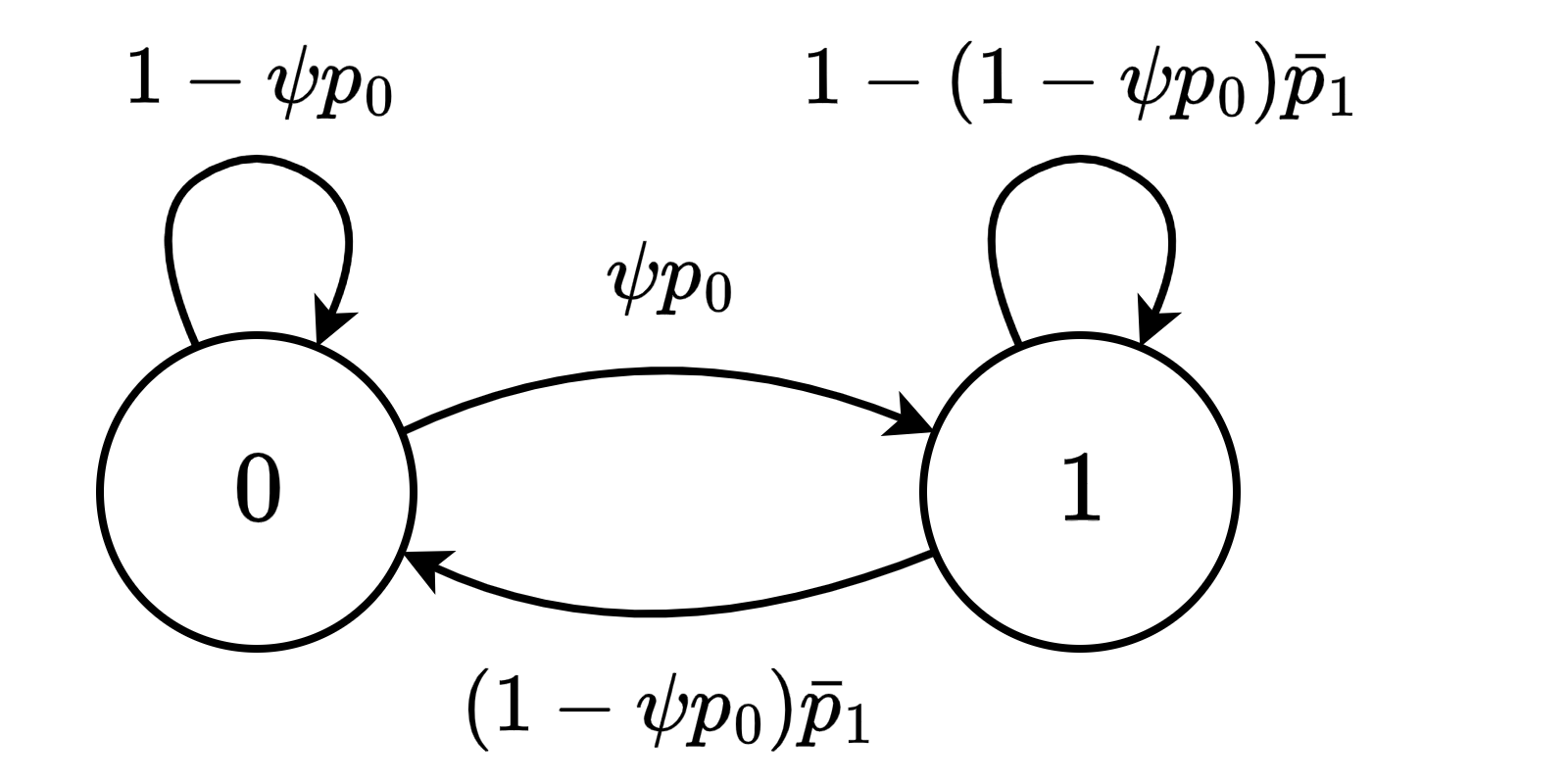}
    \caption{\textsc{DTMC} model of $\mathrm{a}^{\phi_\mathcal{R}}_1$.}
    \label{fig_a1_RS_DTMC}
\end{figure}

\subsection{Analysis of forwarding rates $\{\bar{\eta}^{\phi_{\mathcal{R}}}_i\}_{i=2}^{N}$}

We propose an iterative algorithm to derive the transition probabilities of the state vector $\mathbf{y}_n = \big(\upsigma_1, \mathrm{a}_1, \upsigma_2,\mathrm{a}_2, \dots,\upsigma_n, \mathrm{a}_n\big),\ n \in \{1,2,3,\dots,N\}$, where $\upsigma_{n} \in \{0,1\}$ indicates whether a successful update arrival occurs at node $n$. We can factorize the transition probability as follows:
 \begin{align}
    \label{eqn_IterativeUpdateRate_RS}
     \mathbb{P}(\mathbf{y}_{n+1}^\prime &\mid \mathbf{y}_{n+1})  =\mathbb{P}\left(\mathbf{y}_{n}^\prime, \sigma_{n+1}^\prime, a_{n+1}^\prime \mid \mathbf{y}_{n},\sigma_{n+1},a_{n+1}\right) \notag \\
     &= \mathbb{P}\left(\mathbf{y}_n^\prime \mid \mathbf{y}_n\right) \mathbb{P}\left(\sigma_{n+1}^\prime, a_{n+1}^\prime \mid \mathbf{y}_n^\prime,\mathbf{y}_n,\sigma_{n+1},a_{n+1}\right) \notag \\
     & = \mathbb{P}\left(\mathbf{y}_n^\prime \mid \mathbf{y}_n\right) \mathbb{P}\left(\sigma_{n+1}^\prime \mid \mathbf{y}_n^\prime,\mathbf{y}_n\right) \notag \\ &\times \underbrace{\mathbb{P}\left( a_{n+1}^\prime \mid \mathbf{y}_n^\prime,\mathbf{y}_n,\sigma_{n+1}^\prime,\sigma_{n+1},a_{n+1}\right)}_{K_{\mathbf{y}_n \rightarrow \mathbf{y}_n^\prime}^{\sigma_{n+1} \rightarrow \sigma_{n+1}^\prime} \left( a_{n+1}^\prime \mid a_{n+1}\right)},
 \end{align}
 where the kernels are given as follows:

 \begin{align}
 \label{eqn_TransitionSigma_RS}
     \mathbb{P}\big(\sigma_{n+1}^\prime &\mid \mathbf{y}_n^\prime,\mathbf{y}_n\big) \\
     &= 
     \begin{cases}
         1 & \sigma_{n+1}^\prime = 0, a_n=0, \\
         1 & \sigma_{n+1}^\prime = 1, a_n=1, a_n^\prime=0,\\
         1 & \sigma_{n+1}^\prime = 0, a_n=1, a_n^\prime=1, \sigma_n=0, \\
         1-p_n & \sigma_{n+1}^\prime = 0, a_n=1, a_n^\prime=1, \sigma_n=1, \\
         p_n & \sigma_{n+1}^\prime = 1, a_n=1, a_n^\prime=1, \sigma_n=1.
     \end{cases} \notag
 \end{align}

 \begin{align}
    \label{eqn_TransitionKernel_RS}
     K_{\mathbf{y}_n \rightarrow \mathbf{y}_n^\prime}^{\sigma_{n+1} \rightarrow \sigma_{n+1}^\prime} \!&\!\left( a_{n+1}^\prime \mid a_{n+1}\right) \\
     &= 
     \begin{cases}
         1 & \sigma_{n+1}^\prime = 1, a_{n+1}^\prime=1, \\
         1 & \sigma_{n+1}^\prime = 0, a_{n+1}^\prime=0, a_{n+1}=0, \\
         p_{n+1} & \sigma_{n+1}^\prime = 0, a_{n+1}^\prime=0, a_{n+1}=1, \\
         1-p_{n+1} & \sigma_{n+1}^\prime = 0, a_{n+1}^\prime=1, a_{n+1}=1.
     \end{cases} \notag
 \end{align}

\begin{table}[t]
\centering
\caption{Transition Probabilities for the 2-D \textsc{DTMC} $(\upsigma_1, \mathrm{a}_1)$}
\label{tab_transition_probabilities}
\begin{tabular}{@{}lll@{}}
\toprule
$\mathbf{y}_1=(\upsigma_1, \mathrm{a}_1)$ & $\mathbf{y}_1^\prime=(\upsigma_1', \mathrm{a}_1')$ & \textbf{Probability} $\mathbb{P}(\mathbf{y}_1^\prime | \mathbf{y}_1)$ \\ \midrule
$(0, 0)$ & $(0, 0)$ & $1 - \psi p_0$ \\
         & $(1, 1)$ & $\psi p_0$ \\ \midrule
\multirow{3}{*}{$(0, 1) \text{ or } (1, 1)$} & $(0, 0)$ & $(1 - \psi p_0) p_1$ \\
         & $(0, 1)$ & $(1 - \psi p_0) (1 - p_1)$ \\
         & $(1, 1)$ & $\psi p_0$ \\ \bottomrule
\end{tabular}
\end{table}

Using these kernel functions and starting from $\mathbf{y}_1=\left(\upsigma_1, \mathrm{a}_1\right)$, whose transition probabilities are given in Table \ref{tab_transition_probabilities}, equation~\eqref{eqn_IterativeUpdateRate_RS} can be evaluated iteratively to obtain the transition probabilities of the state vector $\mathbf{y}_n$, which contains all update actions. 
The joint stationary distribution $\pi^{\phi_{\mathcal{R}}}_{\mathbf{y}_N}(\mathbf{y}_N)$, the marginal distributions $\pi^{\phi_{\mathcal{R}}}_{\mathrm{a}_i}(a_i)$, and the update rates $\bar{\eta}^{\phi_{\mathcal{R}}}_i$ are then obtained for $i=1,2,\dots,N$, as presented in the second part of Algorithm~\ref{Alg_IterativeTransitionsMultiHop_RS}.

\begin{algorithm}[t]
\caption{Transition Probability Derivation and Stationary Distribution for Multi-hop \textsc{DTMC} under Policy $(\phi_{\mathcal{R}},\bm{\theta}^{\textsc{fa}})$}
\begin{algorithmic}[1]
\Require Transition probabilities of $\mathbf{y}_1 = (\sigma_1, \mathrm{a}_1)$ and system parameters $\{N, \psi, p_g,p_0,p_1,\dots,p_N\}$
\State \textbf{Initialize:} $\mathbf{y}_1 = (\sigma_1, \mathrm{a}_1)$
\State \textbf{Initialize:} $\mathbb{P}(\mathbf{y}_1^\prime \mid \mathbf{y}_1)$ using Table \ref{tab_transition_probabilities}
\For{$n = 1$ \textbf{to} $N-1$}
    \State Compute $\mathbb{P}\big(\sigma_{n+1}^\prime \mid \mathbf{y}_n^\prime,\mathbf{y}_n\big)$ using equation \eqref{eqn_TransitionSigma_RS}
    \State Compute kernel $K_{\mathbf{y}_n \rightarrow \mathbf{y}_n^\prime}^{\sigma_{n+1} \rightarrow \sigma_{n+1}^\prime} \left( a_{n+1}^\prime \mid a_{n+1}\right)$ using using equation \eqref{eqn_TransitionKernel_RS}
    \State Compute the transition probability $\mathbb{P}(\mathbf{y}_{n+1}^\prime \mid \mathbf{y}_{n+1})$ using equation \eqref{eqn_IterativeUpdateRate_RS} where $\mathbf{y}_{n+1} = \left( \mathbf{y}_n, \upsigma_{n+1}, \mathrm{a}_{n+1} \right)$.
\EndFor
\State \Return $\mathbb{P}(\mathbf{y}_{n}^\prime \mid \mathbf{y}_{n})$ for $n=1, 2, \dots, N$.

\algrule
\raggedright \textbf{Stationary Distribution and Update Rates}
\medskip

\setcounter{ALG@line}{0}
\State Solve $\pi^{\phi_{\mathcal{R}}}_{\mathbf{y}_N}(\mathbf{y}^\prime_N) = \sum_{\mathbf{y}_N} \pi^{\phi_{\mathcal{R}}}_{\mathbf{y}_N}(\mathbf{y}_N) \mathbb{P}(\mathbf{y}^\prime_N \mid \mathbf{y}_N)$ subject to $\sum_{\mathbf{y}_N} \pi^{\phi_{\mathcal{R}}}_{\mathbf{y}_N}(\mathbf{y}_N) = 1$ to obtain $\pi^{\phi_{\mathcal{R}}}_{\mathbf{y}_N}(\mathbf{y}_N)$.
\For{each node $i = 1$ \textbf{to} $N$}
    \State $\pi^{\phi_{\mathcal{R}}}_{\mathrm{a}_i}(a_i) \!=\! \sum_{\sigma_1, \dots, \sigma_N} \sum_{a_{j \neq i}} \pi^{\phi_{\mathcal{R}}}_{\mathbf{y}_N}(\sigma_1, a_1, \dots, \sigma_N, a_N)$
    \State Evaluate the node update rate: $\bar{\eta}^{\phi_{\mathcal{R}}}_i = \pi^{\phi_{\mathcal{R}}}_{\mathrm{a}_i}(1)$
\EndFor
\State \Return $\bar{\eta}^{\phi_{\mathcal{R}}}_i$ for $i=1, 2, \dots, N$.
\end{algorithmic}
\label{Alg_IterativeTransitionsMultiHop_RS}
\end{algorithm}

\section{Numerical Results}

We present the numerical results, first to validate the analytical results derived in the previous sections and, second, to analyze and discuss the behavior of the \textsc{VAoI} and the update rate at the network nodes under rate-constrained update policies at the source and at \textsc{VAoI}-aware or feedback-aware intermediate links. Simulations are conducted over $10^5$ time slots, and the results are averaged over $200$ Monte Carlo iterations to obtain steady-state values. The default parameters considered are $p_g=0.3$, $N=4$, and $\left( p_0,p_1,p_2,p_3,p_4\right)=\left( 0.7,0.8,0.55,0.95,0.8\right)$, unless otherwise stated.

\subsection{VAoI Optimization}

In Fig.~\ref{fig_AvgVAoI_Psi_Policies}, we depict the average \textsc{VAoI} at the network nodes, $\{\bar{\Delta}_n\}_{n=1}^{N}$, as a function of the update rate at the source, $\psi$, under various source update policies. The average \textsc{VAoI} at the first intermediate node, $\bar{\Delta}_1$, is shown directly, whereas for subsequent nodes, $\{\bar{\Delta}_n\}_{n=2}^{N}$, it equals $\bar{\Delta}_1$ plus an offset corresponding to the expected number of version generations during relaying, given by $p_g \sum_{i=1}^{n-1} \frac{1}{p_i}$ (Theorem \ref{Theorem_AvgVAoIlastNode}).

We include two additional baselines alongside the optimal threshold policy (\textsc{VAoI}-Optimal) and the randomized policy: AoI-Optimal and uniform policies. The AoI-Optimal policy is obtained by optimizing the \textsc{AoI} at node $1$; it is a threshold policy whose threshold is derived from \eqref{eqn_OptimalThreshold} by setting $p_g=1$. Under the uniform policy, updates at the source occur periodically every $D=\frac{1}{\psi}$ slots, where $D$ is a positive integer. For the \textsc{VAoI}-optimal (optimal mixed-threshold) and randomized policies, both analytical and simulation curves are shown. The simulation results match and validate the analytical results in \eqref{eqn_OptimalVAoI_Mix} and \eqref{eqn_AvgVAoI1_RS}\footnote{Simulation curves are omitted from the remaining figures to reduce redundancy and improve clarity, as they exactly match the analytical results.}.

Fig.~\ref{fig_AvgVAoI_Psi_Policies} shows that the mixed threshold policy (\textsc{VAoI}-Optimal) outperforms all baseline policies and achieves the lowest \textsc{VAoI} across network nodes. This gain is more pronounced under stricter update rate constraints at the source. When the constraint is relaxed, all policies converge toward always-update behavior and perform similarly, as seen for larger $\psi$ values (e.g., above $0.5$). Among baselines, the AoI-Optimal policy outperforms the uniform policy, which in turn outperforms the randomized policy. Notably, to maintain a target \textsc{VAoI} of $1.5$ at node $1$, the uniform policy reduces the required update rate by $26\%$ (from 0.285 to 0.210), the AoI-optimal policy by $43\%$ (from 0.285 to 0.162), and the \textsc{VAoI}-optimal policy by $54\%$ (to 0.131), compared to the randomized policy. These results highlight that \emph{\textsc{VAoI}-aware optimal policies can significantly reduce transmission rates without sacrificing the conveyed information, thereby improving energy efficiency in IoT networks.}

\textit{Optimal \textsc{VAoI} thresholds}: 
Fig. \ref{fig_Thresholds_ps_pg} illustrates the contours of the optimal threshold $\Delta_\mathcal{T}^\ast$ and the heatmap of $\bar{\Delta}^{\phi^\ast}_1$ over varying success probabilities $p_0$ and version generation probabilities $p_g$ for $\psi = 0.05$. Larger $p_g$ and smaller $p_0$ lead to higher $\bar{\Delta}^{\phi^\ast}_1$ and larger optimal thresholds. According to Theorem \ref{Theorem_OptimalThreshold}, under stringent rate constraints (very low $\psi$), $\Delta_\mathcal{T}^\ast$ can be approximated as $\lceil \frac{p_g}{\psi p_0} \rceil$. Furthermore, when $\psi \geq \frac{p_g}{\alpha_\mathcal{T}}$, it follows that $\Delta_\mathcal{T}^\ast = 1$. In this case, updates occur at all non-zero \textsc{VAoI} states, resulting in the best average \textsc{VAoI} achievable under an unconstrained network, $\bar{\Delta}^{\phi^\ast}_1 = \frac{p_g}{p_0}$.

\begin{figure}[!t]
    \centering
    \includegraphics[width=0.82\linewidth]{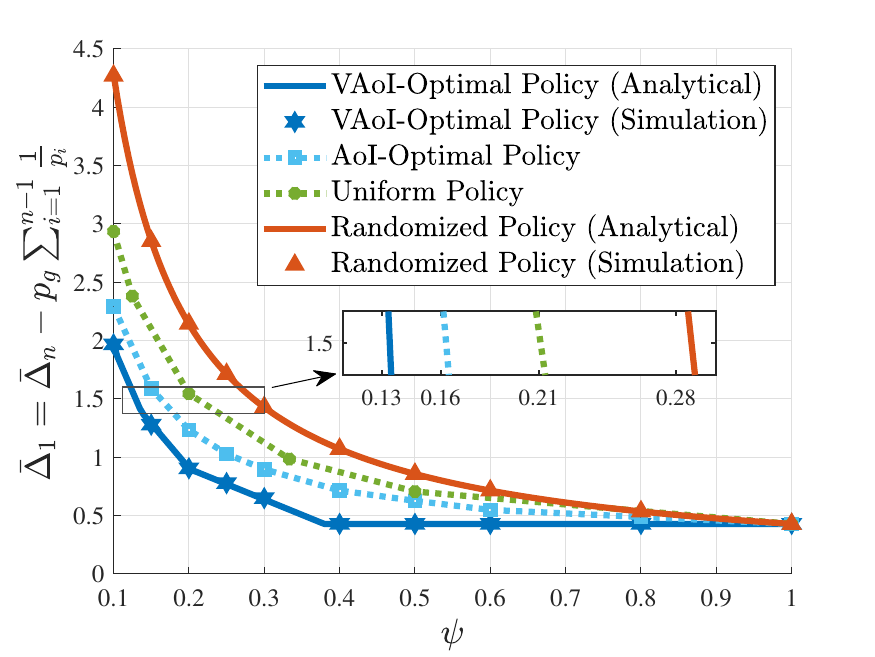}
    \caption{Average \textsc{VAoI} at network nodes vs. $\psi$ under various source policies.}
    \vspace{-8pt}
    \label{fig_AvgVAoI_Psi_Policies}
\end{figure}

    \begin{figure*}[!t]
		\centering
		\begin{minipage}[b]{0.325\linewidth}
            \centering
            \includegraphics[trim={0.5cm 0cm 0.6cm 0cm}, clip,width=\linewidth]{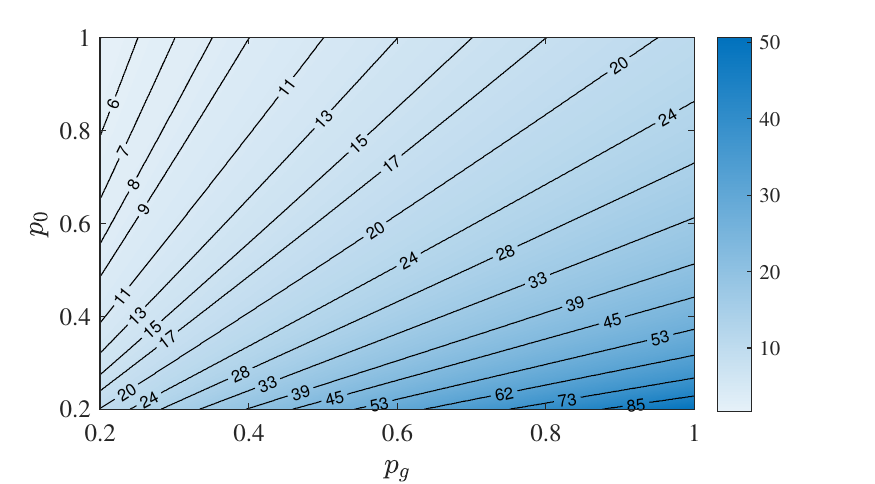}
            \caption{Contour plot of $\Delta^*_{\mathcal{T}}$ and heatmap of $\bar{\Delta}^{\phi^*}_1$ versus $(p_g, p_0)$.}
            \label{fig_Thresholds_ps_pg}
		\end{minipage} \hfill
		\begin{minipage}[b]{0.325\linewidth}
            \centering
            \includegraphics[width=\linewidth]{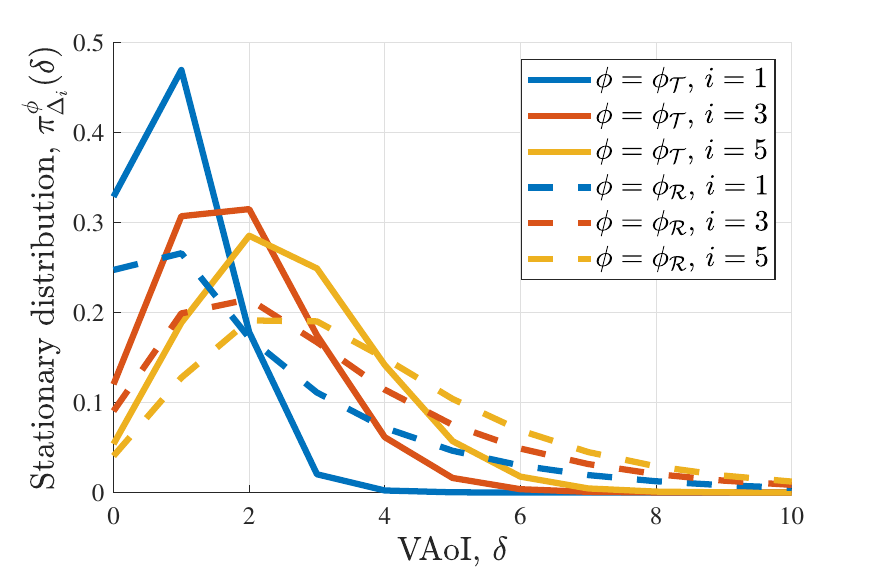}
            \caption{Stationary distribution of \textsc{VAoI} at network nodes for $\Delta_\mathcal{T}=2$ ($\psi=0.201$).}
            \label{fig_DistVAoI_Dt2_phiTphiR}
		\end{minipage} \hfill
		\begin{minipage}[b]{0.325\linewidth}
            \centering
            \includegraphics[width=\linewidth]{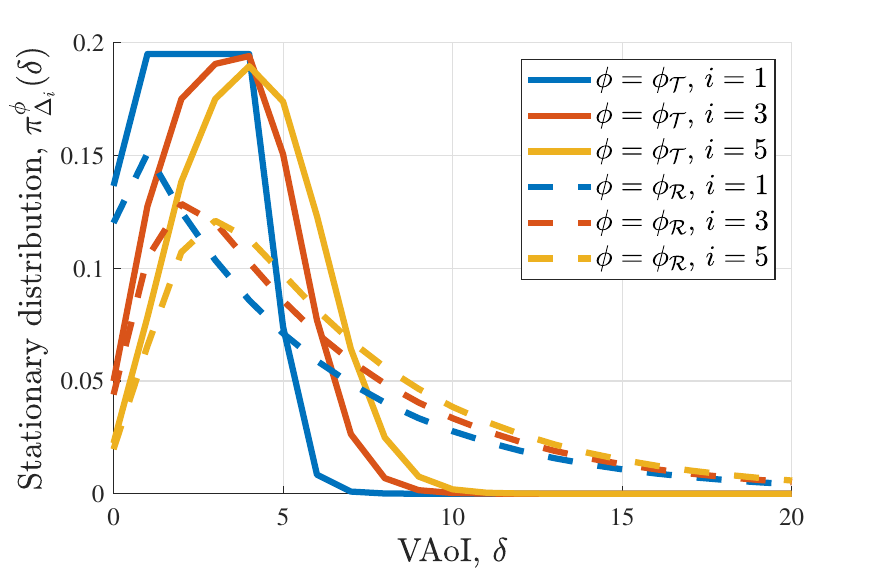}
            \caption{Stationary distribution of \textsc{VAoI} at network nodes for $\Delta_\mathcal{T}=5$ ($\psi=0.084$).}
            \label{fig_DistVAoI_Dt5_phiTphiR}
		\end{minipage}
        \vspace{-8pt}
	\end{figure*}

    \begin{figure*}[!t]
		\centering
		\begin{minipage}[b]{0.325\linewidth}
            \centering
    \includegraphics[width=\linewidth]{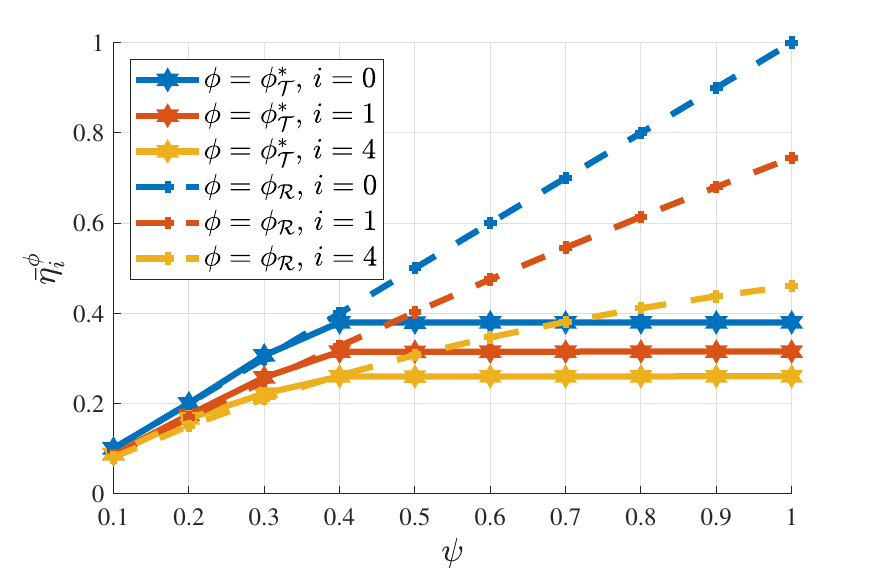}
    \caption{Update rates at network nodes vs. $\psi$ under threshold and randomized policies.}
    \label{fig_UpdateRate_Psi_phiTphiR}
		\end{minipage} \hfill
		\begin{minipage}[b]{0.325\linewidth}
            \centering
            \includegraphics[width=\linewidth]{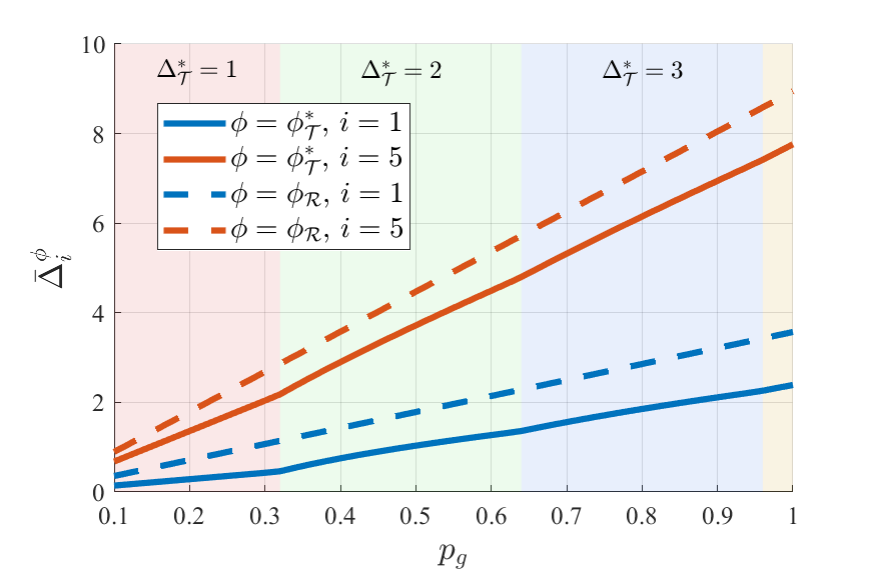}
            \caption{Average \textsc{VAoI} vs. $p_g$ under threshold and randomized policies.}
            \label{fig_AvgVAoI_Pg_phiTphiR}
		\end{minipage} \hfill
		\begin{minipage}[b]{0.325\linewidth}
            \centering
            \includegraphics[width=\linewidth]{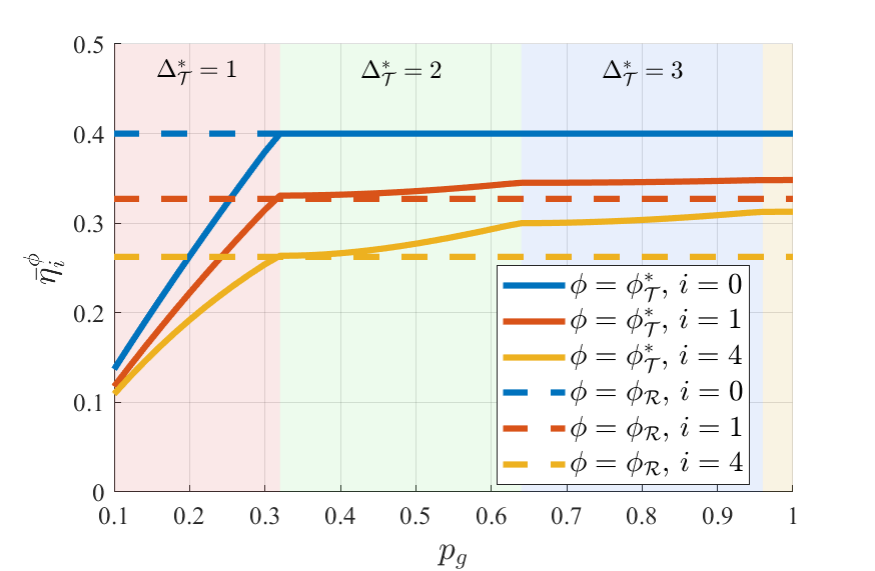}
            \caption{Update rates at network nodes vs. $p_g$ under threshold and randomized policies.}
            \label{fig_AvgRate_Pg_phiTphiR}
		\end{minipage}
        \vspace{-8pt}
	\end{figure*}

\subsection{Stationary Distribution of \textsc{VAoI}}

For further evaluation, we present the stationary distributions of \textsc{VAoI} at network nodes $1$, $3$, and $5$ under the threshold policy with $\Delta_\mathcal{T}=2$ and $\Delta_\mathcal{T}=5$, and the corresponding randomized policies with $\psi = 0.201$ and $\psi=0.084$, as shown in Figs. \ref{fig_DistVAoI_Dt2_phiTphiR} and \ref{fig_DistVAoI_Dt5_phiTphiR}, respectively. In the latter case with stricter rate constraints, the \textsc{VAoI} distribution at node $1$ becomes more dispersed, with higher \textsc{VAoI} values occurring with greater probability. The randomized policy yields a smooth stationary distribution with longer tails, whereas the threshold policy exhibits an almost \emph{uniform} distribution with considerably shorter tails. As stated in Proposition \ref{Prop_StateProbThr}, the \textsc{VAoI} distribution $\pi^{\phi_\mathcal{T}}_{\Delta_1}(\delta)$ under the threshold policy is uniform for $1 \!\leq\! \delta \!<\! \Delta_\mathcal{T}$, decreases by a factor of $\frac{p_g}{\alpha_\mathcal{T}}$ at $\delta = \Delta_\mathcal{T}$, and then decays exponentially at rate $ p_g\bar{p}_0/\alpha_\mathcal{T}$. 
The \textsc{VAoI} distributions at the subsequent nodes $3$ and $5$ are obtained by convolving the distribution of the preceding node with a binomial distribution (see \eqref{eqn_Dist_Di}), resulting in progressively smoother distributions with longer tails, as we proceed from node $1$ to node $5$.

These distributions explain why the optimal threshold policy outperforms the randomized policy: it effectively prevents the \textsc{VAoI} at node $1$ from exceeding the optimal threshold $\Delta_\mathcal{T}$ and keeps the \textsc{VAoI} as uniformly below this threshold as possible, whereas the randomized policy exhibits an exponential decay with a longer tail. Since the average \textsc{VAoI} is a weighted sum of steady-state probabilities, where larger values contribute more, reducing the probability of high \textsc{VAoI} significantly improves performance.

\subsection{Update Rate Optimization}

In Fig. \ref{fig_UpdateRate_Psi_phiTphiR}, the update rates at the source (node $0$) and at the first and last intermediate nodes ($1$ and $4$) are shown for both the optimal mixed-threshold and randomized policies. When $\psi$ is low, i.e., when the update-rate constraint is stringent, updates from the source node are rare. This is achieved by using higher thresholds under the threshold policy and a lower transmission probability $\psi$ under the randomized policy. In this regime, both policies yield similar update rates at the source. The update rates at the intermediate nodes are also similar but remain significantly lower than under the always-update policy with feedback-enabled links, since new updates do not always arrive and retransmissions can be stopped once the last received update has been delivered.

On the other hand, as the rate constraint is relaxed (i.e., as $\psi$ increases), the source update rate under the randomized policy increases linearly. By contrast, under the optimal threshold policy, it reaches a saturated level corresponding to the threshold $\Delta^*_\mathcal{T}=1$. Beyond this point, further transmissions (i.e., $\Delta^*_\mathcal{T}=0$) are unnecessary, since sending updates when the \textsc{VAoI} at node $1$ is zero provides no benefit; therefore, these redundant transmissions are avoided. As a result, this also yields a lower, fixed arrival rate at the subsequent intermediate nodes. We observe that, \emph{although providing feedback at the intermediate nodes under the \textsc{VAoI}-agnostic randomized policy reduces the update rates compared to always-update forwarding, the \textsc{VAoI}-aware source policy achieves a substantially greater reduction.}

Prioritizing fresh and informative data enables version-aware communication to improve network efficiency across diverse resource conditions compared to version-agnostic communication. \emph{Under tight resource constraints, version-aware communication reduces VAoI while utilizing the same amount of resources, while, when resources are abundant, it maintains VAoI while filtering redundant transmissions, thereby minimizing unnecessary resource consumption.}

\subsection{VAoI and Update Rate vs. $p_g$}
We investigate the impact of the version generation probability $p_g$ on the \textsc{VAoI} and the update rate of the network nodes under a fixed source update rate constraint, $\psi=0.4$. In Fig.~\ref{fig_AvgVAoI_Pg_phiTphiR}, the average \textsc{VAoI} of the first intermediate node ($1$) and the destination node ($5$) is shown as a function of $p_g$. We observe, first, that the optimal threshold policy outperforms the randomized policy, with the performance gap increasing as $p_g$ grows. Second, the average \textsc{VAoI} at the destination under both policies increases more steeply with respect to $p_g$ than at node $1$. This additional slope is $\sum_{i=1}^{N-1} \frac{1}{p_i}$, as indicated in \eqref{eqn_AvgVAoIlastNode}.

The corresponding update rates for these two nodes, as well as for the source, are shown in Fig.~\ref{fig_AvgRate_Pg_phiTphiR}. For fixed $\psi=0.4$, we observe that when $p_g$ is low (below $0.32$), the optimal threshold policy at the source can transmit all newly generated versions to node $1$. However, by skipping transmissions when no new version is available, i.e., by setting the threshold $\Delta^*_{\mathcal{T}}=1$, this policy reduces the source update rate compared to the randomized policy, while still achieving lower \textsc{VAoI} at node $1$. This reduction becomes more pronounced as $p_g$ decreases. Consequently, the lower arrival rate at node $1$ results in fewer transmissions at node $1$ and at subsequent nodes up to node $4$ compared to the randomized policy. Note that the update rates under the randomized policy are independent of $p_g$, as it is agnostic to \textsc{VAoI} dynamics.

On the other hand, when $p_g$ increases beyond $0.32$ and $0.64$, the source resources become insufficient to transmit all newly generated versions. In this regime, the optimal policy increases the threshold values to $2$ and $3$, respectively (see \eqref{eqn_OptimalThreshold}). Consequently, the source update rate, and thus the arrival rate at node $1$, becomes equal under both policies (with the mixing coefficient $\kappa$ introduced in Theorem \ref{Theorem_OptimalThreshold}). For the randomized policy, the update rates at intermediate nodes remain constant as $p_g$ increases because it does not account for \textsc{VAoI} evolution. However, the update rates at nodes $1$ and $4$ in Fig.~\ref{fig_AvgRate_Pg_phiTphiR} increase with $p_g$, even though the arrival rate at node $1$ remains unchanged.

This behavior is explained by the fact that, as $p_g$ (and thus $\Delta^*_{\mathcal{T}}$) increases, the optimal threshold policy introduces a \emph{silence} interval after each successful update delivery, during which no transmissions occur. This allows the \textsc{VAoI} at node $1$ to reach the threshold to initiate the next update. During this interval, node $1$ attempts one or a few transmissions (depending on $p_1$) to deliver the update to the next node, and the update then propagates similarly through subsequent nodes. In this regime, updates are less likely to be replaced by newer versions due to the silence period.
In contrast, under the randomized policy, updates are more frequently replaced by newer versions before successful delivery. As a result, some updates are discarded and no longer transmitted, which reduces the average number of transmissions per update. This leads to lower update rates at intermediate nodes than under the threshold policy, albeit with higher VAoI.

\section{Conclusion}
	\label{sec_Conclusion}
	We provided a comprehensive characterization of the \textsc{VAoI} in multi-hop communication networks with transmission constraints and acknowledgment-based feedback. We developed a bi-level optimization framework to jointly address optimal source-side update control and feedback-aware forwarding at intermediate nodes. We showed that the optimal source policy takes a threshold-based form and derived its closed-form expression. For both the optimal threshold policy and a randomized baseline, we obtained closed-form expressions for the stationary distribution, the average \textsc{VAoI}, and the update rates across all network nodes. Our results further quantify the role of feedback in reducing redundant transmissions while preserving the best achievable \textsc{VAoI}. Numerical evaluations corroborate the theoretical analysis and underscore the effectiveness of \textsc{VAoI}-aware and feedback-aware design for efficient multi-hop communication systems.

\bibliographystyle{IEEEtran}
\bibliography{Refs1}

@inproceedings{delfani2025timestamps,
  author={Delfani, Erfan and Pappas, Nikolaos},
  booktitle={Proc. IEEE INFOCOM}, 
  title={From Timestamps to Versions: Version {AoI} in Single- and Multi-Hop Networks}, 
  year={2026},
  volume={},
  number={},
  pages={},
  doi={10.1109/INFOCOM59046.2026.11571750}}

@article{kountouris2021semantics,
  title={Semantics-empowered communication for networked intelligent systems},
  author={Kountouris, Marios and Pappas, Nikolaos},
  journal={IEEE Commun. Mag.},
  volume={59},
  number={6},
  pages={},
  year={2021}
}

@article{yates2021age,
  title={Age of information: An introduction and survey},
  author={Yates, Roy D and Sun, Yin and Brown, D Richard and Kaul, Sanjit K and Modiano, Eytan and Ulukus, Sennur},
  journal={IEEE J. Sel. Areas Commun.},
  volume={39},
  number={5},
  pages={},
  year={2021}
}

@book{altman1999constrained,
  title={Constrained Markov Decision Processes},
  author={Altman, Eitan},
  volume={7},
  year={1999},
  publisher={CRC Press}
}

@article{beutler1985optimal,
  title={Optimal policies for controlled Markov chains with a constraint},
  author={Beutler, Frederick J and Ross, Keith W},
  journal={J. Math. Anal. Appl.},
  volume={112},
  number={1},
  pages={},
  year={1985}
}

@book{norris1998markov,
  title={Markov chains},
  author={Norris, James R},
  number={2},
  year={1998},
  publisher={Cambridge University Press}
}

@book{bertsekas2011dynamic,
  author = {Bertsekas, Dimitri P.},
  title = {Dynamic Programming and Optimal Control, Vol. II},
  year = {2007},
  publisher = {Athena Scientific},
  edition = {3rd}
}

@inproceedings{champati2019distribution,
  title={On the distribution of {AoI} for the {GI/GI/1/1} and {GI/GI/1/2} systems: Exact expressions and bounds},
  author={Champati, Jaya Prakash and Al-Zubaidy, Hussein and Gross, James},
  booktitle={Proc. IEEE INFOCOM},
  year={2019},
  pages={37-45}
}

@article{ji2024age,
  title={Age-optimal packet scheduling with resource constraint and feedback delay},
  author={Ji, Yonghao and Lu, Yuxiao and Xu, Xiaoli and Huang, Xinmei},
  journal={IEEE Trans. Commun.},
  volume={72},
  number={7},
  pages={},
  year={2024}
}

@article{ayan2020probability,
  title={Probability analysis of age of information in multi-hop networks},
  author={Ayan, Onur and G{\"u}rsu, H Murat and Papa, Arled and Kellerer, Wolfgang},
  journal={IEEE Networking Lett.},
  volume={2},
  number={2},
  pages={},
  year={2020}
}

@article{inoue2019general,
  title={A general formula for the stationary distribution of the age of information and its application to single-server queues},
  author={Inoue, Yoshiaki and Masuyama, Hiroyuki and Takine, Tetsuya and Tanaka, Toshiyuki},
  journal={IEEE Trans. Inf. Theory},
  volume={65},
  number={12},
  pages={},
  year={2019}
}

@article{wang2021age,
  title={Age of changed information: Content-aware status updating in the Internet of Things},
  author={Wang, Xijun and Lin, Wenrui and Xu, Chao and Sun, Xinghua and Chen, Xiang},
  journal={IEEE Trans. Commun.},
  volume={70},
  number={1},
  pages={},
  year={2021}
}

@article{jiang2021joint,
  title={Joint performance analysis of ages of information in a multi-source pushout server},
  author={Jiang, Yukang and Miyoshi, Naoto},
  journal={IEEE Trans. Inf. Theory},
  volume={68},
  number={2},
  pages={},
  year={2021}
}

@article{akar2025age,
  title={Age of information in a single-source generate-at-will dual-server status update system},
  author={Akar, Nail and Ulukus, Sennur},
  journal={IEEE Trans. Commun.},
  year={2025},
  volume={73},
  number={9},
  pages={}
}

@article{salimnejad2024age,
  title={Age of information versions: A semantic view of markov source monitoring},
  author={Salimnejad, Mehrdad and Kountouris, Marios and Ephremides, Anthony and Pappas, Nikolaos},
  journal={IEEE Trans. Commun.}, 
  year={2025},
  volume={73},
  number={12},
  pages={}
}

@article{costa2016age,
  title={On the age of information in status update systems with packet management},
  author={Costa, Maice and Codreanu, Marian and Ephremides, Anthony},
  journal={IEEE Trans. Inf. Theory},
  volume={62},
  number={4},
  pages={},
  year={2016}
}

@article{maatouk2020age,
  title={The age of incorrect information: A new performance metric for status updates},
  author={Maatouk, Ali and Kriouile, Saad and Assaad, Mohamad and Ephremides, Anthony},
  journal={IEEE/ACM Trans. Netw.},
  volume={28},
  number={5},
  pages={},
  year={2020}
}

@inproceedings{kaul2012real,
  title={Real-time status: How often should one update?},
  author={Kaul, Sanjit and Yates, Roy and Gruteser, Marco},
  booktitle={Proc. IEEE INFOCOM},
  pages={},
  year={2012}
}

@inproceedings{yates2021vage,
  title={The age of gossip in networks},
  author={Yates, Roy D},
  booktitle={Proc. IEEE ISIT},
  year={2021},
  pages={}
}

@article{chen2024minimizing,
  title={Minimizing age of incorrect information over a channel with random delay},
  author={Chen, Yutao and Ephremides, Anthony},
  journal={IEEE/ACM Trans. Netw.},
  volume={32},
  number={4},
  pages={},
  year={2024}
}

@article{karevvanavar2024version,
  title={Version age of information minimization over fading broadcast channels},
  author={Karevvanavar, Gangadhar and Pable, Hrishikesh and Patil, Om and Bhat, Rajshekhar V and Pappas, Nikolaos},
  journal={IEEE Trans. Wireless Commun.},
  year={2025},
  volume={24},
  number={2},
  pages={}
}

@article{abd2022closed,
  title={Closed-form characterization of the {MGF} of {AoI} in energy harvesting status update systems},
  author={Abd-Elmagid, Mohamed A and Dhillon, Harpreet S},
  journal={IEEE Trans. Inf. Theory},
  volume={68},
  number={6},
  pages={},
  year={2022}
}

@article{inoue2025characterizing,
  title={Characterizing the Age of Information with Multiple Coexisting Data Streams},
  author={Inoue, Yoshiaki and Mandjes, Michel},
  journal={IEEE Trans. Inf. Theory},
  year={2025},
  volume={71},
  number={6},
  pages={}
}

@article{moltafet2022moment,
  title={Moment generating function of age of information in multisource {M/G/1/1} queueing systems},
  author={Moltafet, Mohammad and Leinonen, Markus and Codreanu, Marian},
  journal={IEEE Trans. Commun.},
  volume={70},
  number={10},
  pages={},
  year={2022}
}

@article{fiems2023age,
  title={Age of information analysis with preemptive packet management},
  author={Fiems, Dieter},
  journal={IEEE Commun. Lett.},
  volume={27},
  number={4},
  pages={},
  year={2023}
}

@article{akar2021discrete,
  title={Discrete-time queueing model of age of information with multiple information sources},
  author={Akar, Nail and Dogan, Ozancan},
  journal={IEEE Internet Things J.},
  volume={8},
  number={19},
  pages={},
  year={2021}
}

@article{kosta2021age,
  title={The age of information in a discrete time queue: Stationary distribution and non-linear age mean analysis},
  author={Kosta, Antzela and Pappas, Nikolaos and Ephremides, Anthony and Angelakis, Vangelis},
  journal={IEEE J. Sel. Areas Commun.},
  volume={39},
  number={5},
  pages={},
  year={2021}
}

@inproceedings{zhang2021age,
  title={On age of information for discrete time status updating system with {Ber/G/1/1} queues},
  author={Zhang, Jixiang and Xu, Yinfei},
  booktitle={Proc. IEEE ITW},
  year={2021}
}

@article{yates2020age,
  title={The age of information in networks: Moments, distributions, and sampling},
  author={Yates, Roy D},
  journal={IEEE Trans. Inf. Theory},
  volume={66},
  number={9},
  pages={},
  year={2020}
}

@inproceedings{delfani2024semantics,
  title={Semantics-aware status updates with energy harvesting devices: Query version age of information},
  author={Delfani, Erfan and Pappas, Nikolaos},
  booktitle={Proc. WiOpt},
  year={2024},
  pages={}
}

@article{Chiariotti2021PAoI,
  author={Chiariotti, Federico and Vikhrova, Olga and Soret, Beatriz and Popovski, Petar},
  journal={IEEE Trans. Commun.}, 
  title={Peak Age of Information Distribution for Edge Computing With Wireless Links}, 
  year={2021},
  volume={69},
  number={5}
}

@article{Zhang2025AoIVehicles,
  author={Zhang, Tianci and Chen, Zhengchuan and Tian, Zhong and Wang, Min and Zhen, Li and Wu, Dapeng Oliver and Li, Yonghui and Quek, Tony Q. S.},
  journal={IEEE Trans. Commun.}, 
  title={Age of Information in Internet of Vehicles: A Discrete-Time Multisource Queueing Model}, 
  year={2025},
  volume={73},
  number={5}
}

@article{Tripathi2023Mhop,
  author={Tripathi, Vishrant and Talak, Rajat and Modiano, Eytan},
  journal={IEEE/ACM Trans. Netw.}, 
  title={Information Freshness in Multihop Wireless Networks}, 
  year={2023},
  volume={31},
  number={2}
}

@inproceedings{Vikhrova2020Mhop,
  author={Vikhrova, Olga and Chiariotti, Federico and Soret, Beatriz and Araniti, Giuseppe and Molinaro, Antonella and Popovski, Petar},
  booktitle={Proc. IEEE GLOBECOM}, 
  title={Age of Information in Multi-hop Networks with Priorities}, 
  year={2020}
}

@inproceedings{Talak2018Mhop,
  author={Talak, Rajat and Karaman, Sertac and Modiano, Eytan},
  booktitle={Proc. Allerton Conf.}, 
  title={Minimizing age-of-information in multi-hop wireless networks}, 
  year={2017},
  pages={}
}

@inproceedings{Sinha2024Tandem,
  author={Sinha, Ashirwad and Singhvi, Shubhransh and Mankar, Praful D. and Dhillon, Harpreet S.},
  booktitle={Proc. IEEE ISIT}, 
  title={Peak Age of Information under Tandem of Queues}, 
  year={2024},
  pages={}
}

@article{Bedewy2019Mhop,
  author={Bedewy, Ahmed M. and Sun, Yin and Shroff, Ness B.},
  journal={IEEE/ACM Trans. Netw.}, 
  title={The Age of Information in Multihop Networks}, 
  year={2019},
  volume={27},
  number={3}
}

@article{Chiariotti2022Mhop,
  author={Chiariotti, Federico and Vikhrova, Olga and Soret, Beatriz and Popovski, Petar},
  journal={IEEE Trans. Commun.}, 
  title={Age of Information in Multihop Connections With Tributary Traffic and No Preemption}, 
  year={2022},
  volume={70},
  number={10}
}

@inproceedings{Kaswan2023Mhop,
  author={Kaswan, Priyanka and Ulukus, Sennur},
  booktitle={Proc. IEEE ISIT}, 
  title={Age of Information With Non-Poisson Updates in Cache-Updating Networks}, 
  year={2023},
  pages={}
}

@inproceedings{asvadi2024age,
  title={Age of Information in Multipath Multihop Networks},
  author={Asvadi, Sepehr and Ashtiani, Farid},
  booktitle={Proc. IEEE WCNC},
  year={2024}
}

@article{buyukates2019age,
  title={Age of information in multihop multicast networks},
  author={Buyukates, Baturalp and Soysal, Alkan and Ulukus, Sennur},
  journal={J. Commun. Netw.},
  volume={21},
  number={3},
  pages={},
  year={2019}
}

@article{Delfani2025LEO,
  author={Delfani, Erfan and Pappas, Nikolaos},
  journal={IEEE Commun. Lett.}, 
  title={Semantics-Aware Updates from Remote Energy Harvesting Devices to Interconnected {LEO} Satellites}, 
  year={2025},
  volume={29},
  number={8},
  pages={}
}

@article{Mehrdad2025CL,
  author={Salimnejad, Mehrdad and Pappas, Nikolaos and Kountouris, Marios},
  journal={IEEE Commun. Lett.}, 
  title={So Timely, Yet So Stale: The Impact of Clock Drift in Real-Time Systems}, 
  year={2025},
  volume={29},
  number={10},
  pages={}
}

@article{luo2025survey,
  title={From Information Freshness to Semantics of Information and Goal-oriented Communications},
  author={Luo, Jiping and Delfani, Erfan and Salimnejad, Mehrdad and Pappas, Nikolaos},
  journal={arXiv preprint arXiv:2512.12758},
  year={2025}
}

@article{munari2025s,
  title={What’s My Age of Information Again? The Role of Feedback in {AoI} Optimization Under Limited Transmission Opportunities},
  author={Munari, Andrea and Badia, Leonardo},
  journal={IEEE Trans. Commun.},
  year={2025}
}

@article{lou2021boosting,
  title={Boosting or hindering: {AoI} and throughput interrelation in routing-aware multi-hop wireless networks},
  author={Lou, Jiadong and Yuan, Xu and Kompella, Sastry and Tzeng, Nian-Feng},
  journal={IEEE/ACM Trans. Netw.},
  volume={29},
  number={3},
  pages={},
  year={2021}
}

@inproceedings{feng2022timely,
  title={Timely status update: Should ARQ be used in two-hop networks?},
  author={Feng, Jian and Pan, Haoyuan and Chan, Tse-Tin and Liang, Jiaxin},
  booktitle={Proc. IEEE ICC},
  pages={},
  year={2022}
}

@article{arafa2019timely,
  title={Timely updates in energy harvesting two-hop networks: Offline and online policies},
  author={Arafa, Ahmed and Ulukus, Sennur},
  journal={IEEE Trans. Wireless Commun.},
  volume={18},
  number={8},
  pages={},
  year={2019}
}

@inproceedings{li2020age,
  title={Age-oriented opportunistic relaying in cooperative status update systems with stochastic arrivals},
  author={Li, Bohai and Chen, He and Zhou, Yong ... and Li, Yonghui},
  booktitle={Proc. IEEE GLOBECOM},
  pages={},
  year={2020}
}

@article{liu2021minimizing,
  title={Minimizing {AoI} with throughput requirements in multi-path network communication},
  author={Liu, Qingyu and Zeng, Haibo and Chen, Minghua},
  journal={IEEE/ACM Trans. Netw.},
  volume={30},
  number={3},
  pages={},
  year={2021}
}

@article{ke2025information,
  title={Information freshness in multi-hop satellite IoT systems},
  author={Ke, Ying and Ni, Zihan and Zhang, Di and Miao, Xiaqing and Leow, Chee Yen and Wang, Shuai and Pan, Gaofeng and An, Jianping},
  journal={IEEE Trans. Mobile Comput.},
  year={2025}
}

@article{gu2021optimizing,
  title={Optimizing information freshness in two-hop status update systems under a resource constraint},
  author={Gu, Yifan and Wang, Qian and Chen, He and Li, Yonghui and Vucetic, Branka},
  journal={IEEE J. Sel. Areas Commun.},
  volume={39},
  number={5},
  pages={},
  year={2021}
}

\appendices

\section{Proof of Proposition \ref{Prop_VAoInodei}}
	\label{Appen_Proof_VAoInodei}
	
	\begin{proof}
		Node $i$ transmits updates to node $i+1$ in every time slot. Upon successful reception, node $i+1$ retains only the most recent version, discarding all earlier ones. Thus, the \textsc{VAoI} at node $i+1$ depends on the most recent successful transmissions from node $i$. If the latest transmission at time $t$ succeeds (w.p. $p_i$), node $i+1$ obtains the version held by node $i$ in the previous slot, i.e., $V_{i+1}(t) = V_i(t-1)$. The \textsc{VAoI} $\Delta_{i+1}(t) = V_0(t) - V_{i+1}(t)$ is then:
		\begin{align}
			\label{eqn_ProofVAoI_ip1_first}
			\Delta_{i+1}(t) = \underbrace{V_0(t) - V_0(t\!-\!1)}_{\xi_{1}} + \underbrace{V_0(t\!-\!1) - V_{i}(t-1)}_{\Delta_{i}(t-1)},
		\end{align}
		where $\xi_1 = V_0(t) - V_0(t-1)$ represents the number of new versions generated by the source in the most recent slot (either $0$ or $1$).
		If the latest transmission fails (w.p. $1 - p_i$), but the previous one at $t-1$ succeeds (w.p. $p_i$), then $V_{i+1}(t) = V_i(t-2)$, and the \textsc{VAoI} becomes:
		\begin{align}
			\label{eqn_ProofVAoI_ip1_second}
			\Delta_{i+1}(t) 
			&= \underbrace{V_0(t) - V_0(t\!-\!2)}_{\xi_{2}} + \underbrace{V_0(t\!-\!2) - V_{i}(t\!-\!2)}_{\Delta_{i}(t-2)}, 
		\end{align}
		where $\xi_2$ denotes the number of versions generated over the past two slots. Since the source generates a new version in each slot according to a Bernoulli process with parameter $p_g$, $\xi_k \sim \text{Bin}(k, p_g)$ over $k$ slots. Generally, the \textsc{VAoI} at node $i+1$ at time $t$ equals the \textsc{VAoI} at node $i$ from $m_i$ slots earlier plus the number of generated versions in those $m_i$ slots, where $m_i$ follows a Geometric distribution representing the number of transmissions required for successful delivery over link $i$.
	\end{proof}

	\section{Proof of Lemma \ref{Lemma_VAoIDestNode}}
	\label{Appen_Proof_VAoIDestNode}
    
	\begin{proof}
		According to Proposition \ref{Prop_VAoInodei}, the \textsc{VAoI} at all nodes can be expressed recursively as follows:
		\begin{align*}
        \scalebox{0.8}{$
			\begin{cases}
                \!\Delta_2(t)\!=\!\Delta_1(t\!-\!m_1)\!+\!\xi_{m_{1}}, \\
                \vdots \\
				\!\Delta_N(t)\!=\!\Delta_{N\!-\!1}(t\!-\!m_{N\!-\!1})\!+\!\xi_{m_{N\!-\!1}}, \\
				\!\Delta_{N\!+\!1}(t)\!=\!\Delta_{N}(t\!-\!m_{N})\!+\!\xi_{m_{N}},   
			\end{cases} 
            \hspace{-10pt} \! \Rightarrow \! \Delta_{N\!+\!1}(t)\!=\!\Delta_{1} (t\!-\!\!\underbrace{\sum_{i=1}^{N} m_{i}}_{=\tau_{N}})\!+\!\!\underbrace{\sum_{i=1}^{N} \xi_{m_i}}_{=\beta_{N}}\!.
            $}
		\end{align*}
		
		The expected value of $\tau_N$ and $\beta_N$ is derived:
		\begin{align*}
			\mathbb{E} \left[\tau_N \right] &\!=\!\! \sum_{i=1}^{N} \mathbb{E} \left[ m_{i} \right] \!=\!\! \sum_{i=1}^{N} \frac{1}{p_i}, \\
			\mathbb{E} \left[\beta_N \right] &\!=\! \!\sum_{i=1}^{N} \mathbb{E}  \left[ \xi_{m_i} \right] \!\overset{(a)}{=}\!  \sum_{i=1}^{N} \mathbb{E}_{m_i} \big[ \mathbb{E} \left[ \xi_{m_i} | m_i \right] \big] \!=\! p_g \! \sum_{i=1}^{N} \mathbb{E} \left[  m_i \right],
		\end{align*}
		where $(a)$ follows from the tower rule: $\mathbb{E} \left[ X \right] \!=\! \mathbb{E}_{Y} \!\big[ \mathbb{E}_{X|Y} \left[ X|Y \right] \big]$. Note that $\xi_{m_i} \!\mid\! m_i$ follows a Binomial distribution, as given in \eqref{eq_BinomialPMF}, with mean $m_i p_g$.
	\end{proof}

	\section{Proof of Theorem \ref{Theorem_AvgVAoIlastNode}}
	\label{Appen_Proof_AvgVAoIlastNode}
	
	\begin{proof}
		According to Lemma \ref{Lemma_VAoIDestNode},
		\begin{align}
			\bar{\Delta}_{N+1}(t) & \!=\! \mathbb{E} \left[ \Delta_{N+1}(t) \right] \!=\! \mathbb{E} \left[ \Delta_{1} (t\!-\!\tau_{N})\!+\!\beta_{N} \right]\\
			&\!\overset{(a)}{=}\! {\sum_{\tau = N}^{\infty} \!\mathbb{P}(\tau_N\!=\!\tau) \mathbb{E} \left[ \Delta_{1}(t\!-\!\tau) \right]} \!+\! p_g \!\sum_{i=1}^{N} \frac{1}{p_i},  \notag
		\end{align}
        
		\noindent where $(a)$ follows from the tower rule: 
		$\mathbb{E} \!\left[ \Delta_{1} \!(t\!-\!\tau_{N})\right] \!=\! \mathbb{E}_{\tau_N} \!\big[ \mathbb{E} \left[ \Delta_{1}\!(t\!-\!\tau_N) | \tau_N \right] \big]
			\!\!=\!\!\!\sum_{\tau = N}^{\infty} \!\mathbb{P}(\tau_N\!\!=\!\!\tau) \mathbb{E} \!\left[ \Delta_{1}\!(t\!-\!\tau) \right]\!.$
		The steady-state value $\bar{\Delta}_{N\!+\!1} \!=\! \lim_{t \rightarrow \infty} \bar{\Delta}_{N+1}(t)$ is then:
        
		\begin{align*}
			\bar{\Delta}_{N\!+\!1} \!\!\overset{(b)}{=}\!\! \lim_{t \rightarrow \infty} \!\!\mathbb{E} \!\left[ \Delta_{1}\!(t) \right] \!\!\sum_{\tau = N}^{\infty} \!\!\mathbb{P}(\tau_N\!=\!\tau)  \!+\! p_g \!\!\sum_{i=1}^{N} \!\!\frac{1}{p_i}
			 \!=\! \bar{\Delta}_{1} \!\!+\! p_g \!\!\sum_{i=1}^{N} \!\!\frac{1}{p_i}\!.
		\end{align*}
		
		Equality $(b)$ follows directly from the relation $\lim_{t \rightarrow \infty} \mathbb{E}[\Delta_1(t-\tau)] = \lim_{t \rightarrow \infty} \mathbb{E}[\Delta_1(t)]$, which holds for an ergodic and integrable \textsc{DTMC} $\Delta_1(t)$, where $\mathbb{E}\big[|\Delta_1(t)|\big] < \infty$~\cite[Sec. 1.10]{norris1998markov}. 
	\end{proof}

    \section{Proof of Lemma \ref{Lemma_OptimalCMDP}}
	\label{Sec_CMDP}

	Consider an on--off update policy $\phi$ under which, at each time slot $t$, node $0$ decides whether to transmit ($a^\phi(t) = 1$) or remain idle ($a^\phi(t) = 0$), i.e., $\phi = \big(a^\phi(0), a^\phi(1), \dots \big)$. The objective is to minimize the time-average \textsc{VAoI} at node $1$ subject to the update rate constraint at node $0$. This problem is formulated as a Constrained Markov Decision Process (\textsc{CMDP}):
    
	\begin{align}
		\label{eq_CMDP}
		\min_{\phi \in \Phi} \ \lim_{T\rightarrow\infty} {\frac{1}{T} \mathbb{E} \left[ \sum_{t=0}^{T-1} \Delta^{\phi}_1(t) \Big| s(0) \right]}\!, \ \ \text{subject to}\ \bar{\eta}_0^\phi \le \psi,
	\end{align}
    
    \noindent where $\Phi$ denotes all feasible policies.
    The \textsc{CMDP} is characterized by state $s(t) \in S$, action $a(t) \in A = \{0,1\}$, transition probability $\mathbb{P}(s(t+1) \mid s(t), a(t))$, and transition cost $C(s(t), a(t), s(t+1))$, where, for brevity, the superscript $\phi$ is omitted. The state, $s(t) = \Delta_1(t)$, denotes the \textsc{VAoI} at node $1$; imposing an upper bound $\Delta_{\text{max}}$ yields a finite state space $S = \{0, 1, \dots, \Delta_{\text{max}}\}$. The transition probabilities are:
	\begin{align}
		\label{eqn_CMDPTransProbs}
		\mathbb{P}\!\left(\Delta_1^\prime | \Delta_1,a\right)
		\!=\! 
        \scalebox{0.9}{$
		\begin{cases}
			p_g & a\!=\!0,\  \Delta_1^\prime\!=\!\Delta_1\!+\!1,\ \Delta_1\!\!<\!\!\Delta_{\text{max}},\\
			\bar{p}_g & a\!=\!0,\  \Delta_1^\prime\!=\!\Delta_1,\ \Delta_1\!\!<\!\!\Delta_{\text{max}},\\
			1 & a\!=\!0,\  \Delta_1^\prime\!=\!\Delta_1\!=\!\Delta_{\text{max}},\\
			p_g \bar{p}_0 & a\!=\!1,\  \Delta_1^\prime\!=\!\Delta_1\!+\!1,\ \Delta_1\!\!<\!\!\Delta_{\text{max}},\\
			\bar{p}_g \bar{p}_0  & a\!=\!1,\  \Delta_1^\prime\!=\!\Delta_1,\ \Delta_1\!\!<\!\!\Delta_{\text{max}},\\
			\bar{p}_0 & a\!=\!1,\ \Delta_1^\prime\!=\!\Delta_1\!=\!\Delta_{\text{max}},\\
			p_g p_0 & a\!=\!1,\  \Delta_1^\prime\!=\!1,\\
			\bar{p}_g p_0 & a\!=\!1,\  \Delta_1^\prime\!=\!0.
		\end{cases} 
        $}
	\end{align}
    
    The transition cost at state $s(t)$ under action $a(t)$ is defined as the resulting \textsc{VAoI}, i.e., $C\!\left(s(t), a(t), s(t\!+\!1)\right) \!=\! \Delta_1(t\!+\!1)$.
	The primal \textsc{CMDP} problem \eqref{eq_CMDP} can be reformulated as a Lagrangian dual problem by introducing a multiplier $\lambda \geq 0$:
	\begin{align}
		\label{eq_DualMDP}
		\sup_{\lambda \geq 0} {\min_{\phi \in \Phi} \mathcal{L}(\lambda,\phi)},
	\end{align}
	where $\mathcal{L}(\lambda,\phi)$ denotes the Lagrangian function:
	\begin{align}
		\label{eq_LagFunction}
		\mathcal{L}(\lambda,\phi) \!=\! \lim_{T\rightarrow\infty} {\frac{1}{T} \mathbb{E} \!\left[ \sum_{t=0}^{T-1} \left \{ \Delta_1(t) \!+\! \lambda a(t) \right \} \Big| s(0) \right] \!-\! \lambda \psi}.
	\end{align}
	
	Let $g(\lambda) = \mathcal{L}(\lambda, \phi_\lambda^*)$ denote the dual function, and let $\phi_\lambda^*$ be the policy minimizing $\mathcal{L}(\lambda, \phi)$ for fixed $\lambda$:
	\begin{align}
		\label{eqn_DualMDPpolicy}
		\phi^\ast_\lambda = \argmin_{\phi \in \Phi} \lim_{T\rightarrow\infty} { \frac{1}{T} \mathbb{E} \left[ \sum_{t=0}^{T-1} \left \{ \Delta_1(t) + \lambda a(t) \right \} \Big| s(0) \right] }.
	\end{align}
    
	This corresponds to solving an unconstrained \textsc{MDP} with a modified transition cost:
	\begin{align}
		\label{eq_LagrangianCostFunction}
		C_\lambda \! \left(s(t),a(t),s(t+1)\right) = \Delta_1(t+1) + \lambda a(t).
	\end{align}
	
	For a finite state space $S$, the growth condition in \cite[Eq.~11.21]{altman1999constrained} holds. Since the transition cost $C(s(t), a(t), s(t+1)) \geq 0$ is bounded below, the conditions of \cite[Corollary~12.2]{altman1999constrained} are satisfied, ensuring the optimal solutions of the dual and primal problems coincide. Thus, the optimal solution to the primal \textsc{CMDP} \eqref{eq_CMDP} is found by solving $\sup_{\lambda \geq 0} g(\lambda)$, where $\phi_{\lambda}^\ast$ comes from \eqref{eqn_DualMDPpolicy}. Specifically, the optimal policy is obtained by first solving the unconstrained \textsc{MDP} \eqref{eqn_DualMDPpolicy} for fixed $\lambda$ to get $\phi_{\lambda}^\ast$, and then optimizing $\lambda$ as in \eqref{eq_DualMDP}. We proceed to prove that $\phi_{\lambda}^\ast$ is a threshold policy.
	
	\begin{proposition}
		The optimal policy of the \textsc{MDP} problem \eqref{eqn_DualMDPpolicy} is a threshold policy.
	\end{proposition}
	
	\begin{proof}
		We begin by establishing that the \textsc{MDP} is weakly accessible, thereby ensuring the existence of an optimal policy. An \textsc{MDP} is weakly accessible if its state space can be partitioned into a transient set $S_t$ and a communicating set $S_c$, where all states in $S_c$ are mutually reachable under some stationary policy. For any stationary stochastic policy $\phi$ assigning positive probability to each action $a \in \{0,1\}$, any state $\Delta_1'$ is reachable from $\Delta_1$. Specifically, if $\Delta_1' < \Delta_1$, take $a=1$ once, then $a=0$ for $\Delta_1'$ steps; if $\Delta_1' \geq \Delta_1$, take $a=0$ for $\Delta_1' - \Delta_1$ steps. Hence, the \textsc{MDP} is weakly accessible; thereby by Proposition 4.2.3 in \cite{bertsekas2011dynamic}, the optimal average cost $J_\lambda^*$ is independent of the initial state $s(0)$. Proposition 4.2.6 guarantees the existence of an optimal policy $\phi_\lambda^*$, and Proposition 4.2.1 ensures $J_\lambda^*$, the value function $\mathcal{V}(s)$, and $\phi_\lambda^*$ satisfy the Bellman equations:
		\begin{align}
			\label{eqn_Bellman}
			J_\lambda^*\!\!+\!\mathcal{V}(s)\!=\!\!\!\min_{a\in\left\{0,1\right\}}\!{\!Q_\lambda(s,a)}, \ \  
			\phi^\ast(s) \!\in\!  \argmin_{a\in\left\{0,1\right\}}{Q_\lambda(s,a)},
		\end{align}
		where $Q_\lambda(s,a) = C_\lambda(s,a)+\!\sum_{s^\prime\in S}{\mathbb{P}\!\left(s^\prime \big | s,a\right)\!\mathcal{V}(s^\prime)}.$
		Here, $C_\lambda(s,a)$ represents the average cost per slot, defined by the transition costs as: $C_\lambda(s,a) = \sum_{s^\prime\in S} {\mathbb{P}\!\left(s^\prime \big | s,a\right) C_\lambda\left(s,a,s^\prime\right)},$ with $C_\lambda(s,a,s^\prime) = \Delta_1^\prime + \lambda a$.
		The Bellman equation for state $s = \Delta_1$ can be written as $a^\ast(\Delta_1) = 1$ if $Q_\lambda(\Delta_1, 1) < Q_\lambda(\Delta_1, 0)$, and $a^*(\Delta_1) = 0$ otherwise. Thus, the optimal action $a^*(\Delta_1)$ depends on the sign of the difference $\partial\mathcal{V}(\Delta_1) = Q_\lambda(\Delta_1, 1) - Q_\lambda(\Delta_1, 0)$. 
		We next show that $\partial \mathcal{V}(\Delta_1)$ is a decreasing function of $\Delta_1$. For $\Delta^{\!-}_1 \leq \Delta^{\!+}_1$, we prove that $\partial \mathcal{V}(\Delta^{\!+}_1) \leq \partial \mathcal{V}(\Delta^{\!-}_1)$, i.e., 
		\begin{align}
			\label{eq_DVineq}
			\partial \mathcal{V}(\Delta^{\!+}_1) - \partial \mathcal{V}(\Delta^{\!-}_1) \leq 0.
		\end{align}
		
		This implies a threshold policy: if $\partial \mathcal{V}(\Delta_\mathcal{T}) < 0$ for some $\Delta_\mathcal{T}$, then $\partial \mathcal{V}(\Delta_1) < 0$ for all $\Delta_1 \geq \Delta_\mathcal{T}$, so the optimal action remains $1$ for all such states.		
		Using \eqref{eqn_CMDPTransProbs}, we obtain:
		\begin{align*}
			\partial \mathcal{V}(\Delta^{\!+}_1) &\!-\! \partial \mathcal{V}(\Delta^{\!-}_1) 
			\!=\! {-p_0 \big[ \Delta^{\!+}_1 \!-\! \Delta^{\!-}_1 \big]} \!-\! \bar{p}_g p_0 \big[ \mathcal{V}(\Delta^{\!+}_1) \!-\! \mathcal{V}(\Delta^{\!-}_1)\big] \\
			& \!-\!\bar{p}_g p_0 \big[ \mathcal{V}(\Delta^{\!+}_1 \!+\! 1) \!-\! \mathcal{V}(\Delta^{\!-}_1 \!+\! 1)\big].
		\end{align*}
		
		The first term is non-positive. Thus, to prove inequality \eqref{eq_DVineq}, it suffices to prove that $\mathcal{V}(\Delta_1)$ is increasing in $\Delta_1$, i.e., for $\Delta^{\!-}_1 \leq \Delta^{\!+}_1$, 
		$\mathcal{V}(\Delta^{\!-}_1) \leq \mathcal{V}(\Delta^{\!+}_1).$
		We use the Value Iteration Algorithm (\textsc{VIA}) and induction. \textsc{VIA} converges to $V(\Delta_1)$ regardless of the initial $\mathcal{V}_0(\Delta_1)$, i.e., $\lim_{k\to\infty} \mathcal{V}_k(\Delta_1) = \mathcal{V}(\Delta_1)$ for all $\Delta_1 \in S$.
		The \textsc{VIA} iteration is:
		\begin{align*}
			\mathcal{V}_{k+1}(\Delta_1)\!=\!\!\!\min_{a\in\left\{0,1\!\right\}} \!\!\bigg\{\!\underbrace{\!\sum_{\Delta_1^\prime\in S} \!\!\mathbb{P}\!\left(\Delta_1^\prime \big|\Delta_1,a\right) \!\Big( \Delta_1^\prime \!+\! \lambda a \!+\! \mathcal{V}_k(\Delta_1^\prime) \Big)}_{Q_{\lambda,k}(\Delta_1,a)}\!\!\bigg\}.
		\end{align*}
		
		We prove by induction that $\mathcal{V}_k(\Delta^{\!-}_1) \!\leq\! \mathcal{V}_k(\Delta^{\!+}_1)$ for all $k \!\geq\! 0$. For $k\!=\!0$, $\mathcal{V}_0(\Delta_1)\!=\!0$, so the claim holds. Assume $\mathcal{V}_k(\Delta^{\!-}_1) \!\leq\! \mathcal{V}_k(\Delta^{\!+}_1)$ and show $\mathcal{V}_{k+1}(\Delta^{\!-}_1) \!\leq\! \mathcal{V}_{k+1}(\Delta^{\!+}_1)$. Since $\mathcal{V}_{k+1}(\Delta_1) = \min \{\mathcal{V}_{k+1}^0(\Delta_1), \mathcal{V}_{k+1}^1(\Delta_1)\}$ with $\mathcal{V}_{k+1}^0 \!=\! Q_{\lambda,k}(\Delta_1,0)$ and $\mathcal{V}_{k+1}^1 \!=\! Q_{\lambda,k}(\Delta_1,1)$, it suffices to prove $\mathcal{V}_{k+1}^0(\Delta^{\!-}_1) \!\leq\! \mathcal{V}_{k+1}^0(\Delta^{\!+}_1)$ and $\mathcal{V}_{k+1}^1(\Delta^{\!-}_1) \!\leq\! \mathcal{V}_{k+1}^1(\Delta^{\!+}_1)$, since then $\min\{\mathcal{V}_{k+1}^0(\Delta^{\!-}_1),\mathcal{V}_{k+1}^1(\Delta^{\!-}_1)\} \!\leq\! \min\{\mathcal{V}_{k+1}^0(\Delta^{\!+}_1),\mathcal{V}_{k+1}^1(\Delta^{\!+}_1)\}$. By the induction hypothesis, all bracketed terms below are non-positive:
		\begin{align*}
			\mathcal{V}&^0_{\!k+1}(\Delta^{\!-}_1) \!-\! \mathcal{V}^0_{\!k+1}(\Delta^{\!+}_1) \!=\! \big[\Delta^{\!-}_1 \!-\! \Delta^{\!+}_1 \big] \\ 
			&\!+\!  \bar{p}_g \big[ \mathcal{V}_{\!k}(\Delta^{\!-}_1) \!-\! \mathcal{V}_{\!k}(\Delta^{\!+}_1) \big] \!+\! p_g \big[\mathcal{V}_{\!k}(\Delta^{\!-}_1 \!\!+\! 1) \!-\! \mathcal{V}_{\!k}(\Delta^{\!+}_1 \!\!+\! 1) \big].  \\
			\mathcal{V}&^1_{\!k+1}(\Delta^{\!-}_1) \!-\! \mathcal{V}^1_{\!k+1}(\Delta^{\!+}_1) \!=\! \bar{p}_0\big[\Delta^{\!-}_1 \!\!-\! \Delta^{\!+}_1 \big] \\ 
			&\!+\!  \bar{p}_g \bar{p}_0 \big[ \mathcal{V}_{\!k}(\Delta^{\!-}_1) \!-\! \mathcal{V}_{\!k}(\Delta^{\!+}_1) \big] \!+\! p_g \bar{p}_0 \big[\mathcal{V}_{\!k}(\Delta^{\!-}_1 \!\!+\! 1) \!-\! \mathcal{V}_{\!k}(\Delta^{\!+}_1 \!\!+\! 1) \big]. 
		\end{align*}
        Thus $\mathcal{V}_{k+1}(\Delta^{\!-}_1) \!\leq\! \mathcal{V}_{k+1}(\Delta^{\!+}_1)$, completing the proof.
	\end{proof}

    \section{Proof of Proposition \ref{Prop_StateProbThr}}
	\label{Appen_Proof_ThrPolicy_SSProbs}
	
	\begin{proof}
		For $\Delta_\mathcal{T} = 0$, when the \textsc{VAoI} at node $1$ is zero, transmission decisions do not affect state transitions, since the \textsc{VAoI} depends solely on whether a new version is generated. Hence, the system dynamics and steady-state distribution are identical for $\Delta_\mathcal{T} = 0$ and $\Delta_\mathcal{T} = 1$. The case $\Delta_\mathcal{T} = 0$ represents an \emph{always-update} policy, which is equivalent to a randomized policy with $\psi = 1$ (see Proposition \ref{Prop_StateProbRandom} and Appendix \ref{Appen_Proof_RSpolicy_SSProbs}).
		For $\Delta_\mathcal{T} \geq 2$, if $\Delta_1 < \Delta_\mathcal{T}$, no transmission occurs; the \textsc{VAoI} increments by $1$ if a new version is generated, otherwise it remains the same. If $\Delta_1 \geq \Delta_\mathcal{T}$, transmission occurs: on success, \textsc{VAoI} resets to $1$ or $0$ depending on whether a new version is generated; on failure, it increases by $1$ if a new version is generated, or remains unchanged otherwise. The resulting Markov chain is shown in Fig.~\ref{fig_MC_ThrPolicyDTo2}, with balance equations provided for $\Delta_\mathcal{T} \geq 2$:
		\[
			\begin{array}{ll}
				\pi^{\phi_{\mathcal{T}}}_{\Delta_1}(0) \!=\! \bar{p}_g\pi^{\phi_{\mathcal{T}}}_{\Delta_1}(0) \!+\! \bar{p}_g p_0 {\sum_{i=\Delta_\mathcal{T}}^{\infty} \! \pi^{\phi_{\mathcal{T}}}_{\Delta_1}(i)}, \\
				\pi^{\phi_{\mathcal{T}}}_{\Delta_1}(1) \!=\! p_g \pi^{\phi_{\mathcal{T}}}_{\Delta_1}(0) \!+\! \bar{p}_0 \pi^{\phi_{\mathcal{T}}}_{\Delta_1}(1) \!+\! p_g p_0 {\sum_{i=\Delta_\mathcal{T}}^{\infty} \! \pi^{\phi_{\mathcal{T}}}_{\Delta_1}(i)}, \\
				\pi^{\phi_{\mathcal{T}}}_{\Delta_1}(\delta_1) \!=\! p_g \pi^{\phi_{\mathcal{T}}}_{\Delta_1}(\delta_1\!-\!1) \!+\! \bar{p}_g \pi^{\phi_{\mathcal{T}}}_{\Delta_1}(\delta_1\!-\!1), \hfill 2 \!\leq\! \delta_1 \!<\! \Delta_\mathcal{T}, \\
				\pi^{\phi_{\mathcal{T}}}_{\Delta_1}(\Delta_\mathcal{T}) \!=\! p_g \pi^{\phi_{\mathcal{T}}}_{\Delta_1}(\Delta_\mathcal{T}-1) \!+\! \bar{p}_g \bar{p}_0 \pi^{\phi_{\mathcal{T}}}_{\Delta_1}(\Delta_\mathcal{T}), \\
				\pi^{\phi_{\mathcal{T}}}_{\Delta_1}(\delta_1) \!=\! p_g \bar{p}_0 \pi^{\phi_{\mathcal{T}}}_{\Delta_1}(\delta_1\!-\!1) \!+\! \bar{p}_g \bar{p}_0 \pi^{\phi_{\mathcal{T}}}_{\Delta_1}(\delta_1\!-\!1), \quad \hfill \delta_1 \!>\! \Delta_\mathcal{T}.
			\end{array}
		\]
        
		The third line is omitted when $\Delta_\mathcal{T} \!=\! 2$. By applying $\sum_{\delta_1=0}^{\infty} \pi^{\phi_{\mathcal{T}}}_{\Delta_1}(\delta_1) \!=\! 1$  and simplifying, Proposition~\ref{Prop_StateProbThr} follows.
	\end{proof}
    
        \begin{figure}[tb]
			\centering
			\includegraphics[trim={0.5cm 0cm 0.5cm 0cm}, clip,scale=0.44]{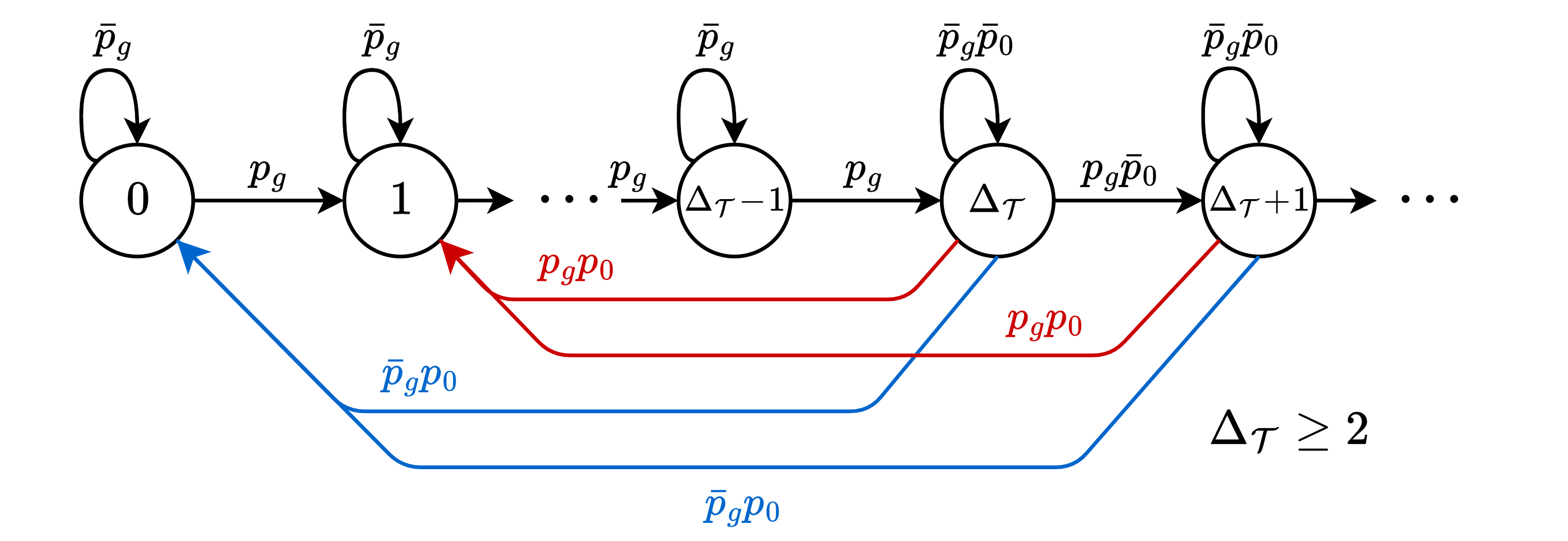} 
			\caption{\textsc{DTMC} model of $\Delta^{\phi_\mathcal{T}}_1$ for $\Delta_\mathcal{T} \!\geq\! 2$.}
			\label{fig_MC_ThrPolicyDTo2}
		\end{figure}

        \section{Proof of Lemma \ref{Lemma_AvgD1Eta0}}
        \label{Appen_Proof_AvdD1Eta1Thr}
        \begin{proof}
		For $\Delta_\mathcal{T} = 0$ and $\Delta_\mathcal{T} = 1$, the average \textsc{VAoI} is identical, as shown in the proof of Proposition \ref{Prop_StateProbThr}, and equals the average \textsc{VAoI} in Proposition \ref{Prop_StateProbRandom} with $\psi = 1$, i.e., $\bar{\Delta}^{\phi_\mathcal{T}(\Delta_\mathcal{T}=0)}_{1} \!=\! \bar{\Delta}^{\phi_\mathcal{T}(\Delta_\mathcal{T}=0)}_{1} \!=\! \frac{p_g}{p_0}$. For $\Delta_\mathcal{T} \geq 2$, using the steady-state probabilities from Proposition \ref{Prop_StateProbThr}, the expected \textsc{VAoI}, $\bar{\Delta}^{\phi_\mathcal{T}}_{1} \!=\! \sum_{\delta_1=0}^\infty \delta_1 \pi^{\phi_{\mathcal{T}}}_{\Delta_1}(\delta_1)$, is given by:
		\begin{align*}
        \scalebox{0.96}{$
			\bar{\Delta}^{\phi_\mathcal{T}}_{1} 
			\!\!\!=\! \pi^{\phi_{\mathcal{T}}}_{\Delta_1} \!\!\left(\Delta_\mathcal{T}\!-\!1\right) \!\Bigg\{\!\! \sum_{\delta_1=1}^{\Delta_\mathcal{T}\!-\!1} \!\! \delta_1
			\!+\! \frac{p_g}{\alpha_\mathcal{T}} \Delta_\mathcal{T} \!+\! 
			\! \frac{p_g}{\alpha_\mathcal{T}} \!\!\left[ \! \frac{r}{(1\!-\!r)^2} \!+\! \Delta_\mathcal{T} \frac{r}{1\!-\!r} \!\right]  \!\! \Bigg\}, $} 
		\end{align*}
		where $r = \frac{p_g \bar{p}_0}{\alpha_\mathcal{T}}$, and after some algebraic manipulation, the final expression for the average \textsc{VAoI} \eqref{eqn_AvgD1Threshold} is obtained.

        For a threshold policy with parameter $\Delta_\mathcal{T}$, $\bar{\eta}^{\phi_{\mathcal{T}}}_0$ represents the probability of the states that trigger transmission, i.e., $\bar{\eta}^{\phi_{\mathcal{T}}}_0 \!=\! \mathbb{P}(\Delta^{\phi_{\mathcal{T}}}_1 \!\geq\! \Delta_\mathcal{T})$, given by:
		\begin{align}
            \label{eqn_UpdateRateNode0_Thr}
			\bar{\eta}^{\phi_{\mathcal{T}}}_0 \!=\!\!\sum_{\delta_1=\Delta_\mathcal{T}}^{\infty} \!\! \pi^{\phi_{\mathcal{T}}}_{\Delta_1}(\delta_1) 
			\!\overset{(a)}{=}\!
			\begin{cases}
				1, & \Delta_\mathcal{T} \!=\! 0, \\
				\frac{p_g}{(\Delta_\mathcal{T}\!-\!1)p_0+\alpha_\mathcal{T}}, & \Delta_\mathcal{T} \!\geq\! 1,
			\end{cases}
		\end{align}
		where step $(a)$ follows directly from the balance equations and the steady-state probabilities presented in Proposition~\ref{Prop_StateProbThr}, using a geometric series analysis.
	\end{proof}

    \section{Proof of Theorem \ref{Theorem_OptimalThreshold}}
	\label{Appen_Proof_OptimalThr}
	
	\begin{proof}
		Threshold policies that satisfy the update rate constraint are considered feasible. For a threshold policy with parameter $\Delta_\mathcal{T}$, the update rate at the source node $f(\Delta_\mathcal{T}) \!=\! \bar{\eta}^{\phi_{\mathcal{T}}}_0$, given by \eqref{eqn_UpdateRateNode0_Thr}:
		\begin{align}
            \label{eqn_FeasibleThresholds}
			f(\Delta_\mathcal{T}) \!=\! \mathbb{P}(\Delta^{\phi_{\mathcal{T}}}_1 \!\geq\! \Delta_\mathcal{T}) \!=\!
				\frac{p_g}{(\Delta_\mathcal{T}\!-\!1)p_0 \!+\! \alpha_\mathcal{T}}, \ \ \hfill \Delta_\mathcal{T} \!\geq\! 1.
		\end{align}
		
        Therefore, a threshold policy is feasible if it satisfies $f(\Delta_\mathcal{T}) \leq \psi$, which can be simplified using \eqref{eqn_FeasibleThresholds}:
		\begin{align}
			\label{eqn_ThresholdLowerBound}
			\Delta_\mathcal{T} \geq \frac{p_g}{p_0}\left( \frac{1}{\psi} - 1 +p_0\right)\!.
		\end{align}
		
		This implies that $\Delta_\mathcal{T}$ must exceed a certain lower bound. Meanwhile, Lemma~\ref{Lemma_AvgD1Eta0} shows that the average \textsc{VAoI} under the threshold policy \eqref{eqn_AvgD1Threshold} is an increasing function of $\Delta_\mathcal{T}$, since it can be rewritten as:
		$\bar{\Delta}^{\phi_{\mathcal{T}}}_1 \!=\! \frac{\Delta_\mathcal{T}}{2} \left( 1\!-\!\frac{\alpha_\mathcal{T}}{(\Delta_\mathcal{T}\!-\!1)p_0 \!+\! \alpha_\mathcal{T}} \right) \!+\! \frac{p_g}{p_0}.$
		The constant term $\frac{p_g}{p_0}$ remains fixed; the first term increases with $\Delta_\mathcal{T}$ due to the linear growth of $\frac{\Delta_\mathcal{T}}{2}$ and the rising value of $\left( 1 \!-\! \frac{\alpha_\mathcal{T}}{(\Delta_\mathcal{T} - 1)p_0 + \alpha_\mathcal{T}} \right)$, since $\frac{\alpha_\mathcal{T}}{(\Delta_\mathcal{T} - 1)p_0 + \alpha_\mathcal{T}}$ decreases as $\Delta_\mathcal{T}$ grows. Therefore, minimizing the average \textsc{VAoI} under the update rate constraint entails selecting the smallest integer $\Delta_\mathcal{T}$ that satisfies \eqref{eqn_ThresholdLowerBound}. However, exact equality may not always be attainable, as $\Delta_\mathcal{T}$ must be integer-valued.
		To achieve $f(\Delta_\mathcal{T}) \!=\! \psi$, a randomized mixture policy can be employed~\cite{beutler1985optimal}\cite[Sec.~6.3]{altman1999constrained}, combining two thresholds $\Delta_\mathcal{T}^\ast$ and $\Delta_\mathcal{T}^\ast \!-\! 1$, where $f(\Delta_\mathcal{T}^\ast) \!\leq\! \psi$ and $f(\Delta_\mathcal{T}^\ast \!-\! 1) \!>\! \psi$, resulting in \eqref{eqn_OptimalThreshold}. The threshold $\Delta_\mathcal{T}^\ast$ is applied w.p. $\kappa$ and $\Delta_\mathcal{T}^\ast \!-\! 1$ with $1 \!-\! \kappa$, where $\kappa f(\Delta_\mathcal{T}^\ast) + (1-\kappa) f(\Delta_\mathcal{T}^\ast-1) = \psi$, leading to \eqref{Optimal_gamma}. This \emph{mixed threshold policy} constitutes the optimal solution to the \textsc{CMDP} under the average rate constraint \eqref{eq_constraint} \cite{beutler1985optimal}.
	\end{proof}

    \section{Proof of Theorem \ref{Theorem_RecDistVAoIi}}
    \label{Appen_Proof_RecDistVAoIi}

    \begin{proof}
    By Proposition~\ref{Prop_VAoInodei}, the recursive update of the \textsc{VAoI} at node $i+1$ is given by $\Delta^{\phi}_{i+1}(t) = \Delta^{\phi}_{i}(t-m_i) + \xi_{m_i}$. In the stationary regime ($t \to \infty$), $\Delta^{\phi}_i$ becomes time-invariant in distribution. Let $\pi^{\phi}_{\Delta_{i}}(\delta)$ and $\pi^{\phi}_{\Delta_{i+1}}(\delta)$ denote the stationary distributions of the \textsc{VAoI} at nodes $i$ and $i+1$, respectively. Conditioning on the geometric delay $m_i = \ell$ and applying the law of total probability, we obtain:
\begin{align}
    \pi^{\phi}_{\Delta_{i+1}}\!(\delta) \!=\!\! \sum_{\ell=1}^{\infty} \mathbb{P}(m_i \!=\! \ell)
     \mathbb{P}\big(\Delta^{\phi}_{i}(t\!-\!\ell) \!+\! \xi_{m_i} \!\!=\! \delta \!\mid\! m_i \!=\! \ell\big).
    \label{proof_step1}
\end{align}

Given $m_i = \ell$, the random variable $\xi_{m_i} \mid m_i=\ell$ follows a Binomial distribution and is independent of $\Delta_{i}(t-\ell)$. Hence, since the PMF of the sum of independent random variables is given by convolution, the conditional probability in \eqref{proof_step1} can be written using the convolution operator $\circledast$ as:
\begin{align*}
    \mathbb{P}\!\left(\Delta^{\phi}_{i}(t\!-\!\ell) \!+\! \xi_{m_i} \!=\! \delta \mid m_i \!=\! \ell\right)
    &= \mathcal{B}(\delta; \ell, p_g) \circledast \pi^{\phi}_{\Delta_{i}}(\delta).
\end{align*}

Substituting this expression into \eqref{proof_step1} yields the first equality of the theorem. Expanding the discrete convolution over the support of the Binomial distribution, $r \in \{0, 1, \dots, \ell\}$, gives:
\begin{align}
    \mathcal{B}(\delta; \ell, p_g) \circledast \pi^{\phi}_{\Delta_{i}}\!(\delta) \!=\!\!
    \sum_{r=0}^{\ell}
    \mathbb{P}(\xi_{m_i}\!=\!r \!\mid\! m_i\!=\!\ell)
    \pi^{\phi}_{\Delta_{i}}\!(\delta\!-\!r).
\end{align}

Substituting this back into \eqref{proof_step1} completes the proof.
    \end{proof}

    \section{Proof of Theorem \ref{Theorem_AvgEta1_Thr}}
    \label{Appen_Proof_AvgEta1_Thr}

    \begin{proof}

    The transition from $\Delta^{\phi_\mathcal{T}}_1(t)=\delta_1$ to $\Delta^{\phi_\mathcal{T}}_1(t+1)=\delta_1^\prime$, i.e., $\mathbb{P}(\delta_1^\prime \mid \delta_1)$, is independent of $\mathrm{a}^{\phi_\mathcal{T}}_1(t)$. Therefore, for the transition from $(\delta_1,a_1)$ to $(\delta_1^\prime,a_1^\prime)$, we have the following:
    \begin{align}
        \mathbb{P}(\Delta_1^\prime,a_1^\prime \mid \Delta_1,a_1) = \mathbb{P}(\Delta_1^\prime \mid \Delta_1)\mathbb{P}(a_1^\prime \mid \Delta_1^\prime,\Delta_1,a_1).
    \end{align}

Here, the superscript $\phi_\mathcal{T}$ is omitted for notational simplicity. From the \textsc{DTMC} of $\mathbf{v}=(\Delta_1,\mathrm{a}_1)$ depicted in Fig. \ref{fig_D1a1_DTMC}:

\begin{align*}
    \pi^{\phi_{\mathcal{T}}}_\mathbf{v}\!(0,1) & \!=\! \bar{p}_g \bar{p}_1 \pi^{\phi_{\mathcal{T}}}_\mathbf{v}\!(0,1) \!+\!\!\!\! \sum_{\delta \geq \Delta_\mathcal{T}} \!\!\! \bar{p}_g p_0 \! \left[\pi^{\phi_{\mathcal{T}}}_\mathbf{v}\!(\delta,0) \!+\! \pi^{\phi_{\mathcal{T}}}_\mathbf{v}\!(\delta,1) \right] \\
    &\!=\!\bar{p}_g \bar{p}_1 \pi^{\phi_{\mathcal{T}}}_\mathbf{v}(0,1) \!+\! \bar{p}_g p_0 \bar{\eta}^{\phi_{\mathcal{T}}}_0, \\
    \pi^{\phi_{\mathcal{T}}}_\mathbf{v}\!(1,1) &\!=\! p_g \bar{p}_1 \pi^{\phi_{\mathcal{T}}}_\mathbf{v}(0,1) \!+\! \bar{p}_g \bar{p}_1 \pi^{\phi_{\mathcal{T}}}_\mathbf{v}(1,1) + p_g p_0 \bar{\eta}^{\phi_{\mathcal{T}}}_0, \\
    \pi^{\phi_{\mathcal{T}}}_\mathbf{v}\!(\delta,1) &\!=\! p_g \bar{p}_1 \pi^{\phi_{\mathcal{T}}}_\mathbf{v}(\delta\!-\!1,1) \!+\! \bar{p}_g \bar{p}_1 \pi^{\phi_{\mathcal{T}}}_\mathbf{v}(\delta,1), \ \hfill 2 \!\leq\! \delta \!<\! \Delta_\mathcal{T}, \\
    \pi^{\phi_{\mathcal{T}}}_\mathbf{v}\!(\Delta_\mathcal{T},1) &\!=\! p_g \bar{p}_1 \pi^{\phi_{\mathcal{T}}}_\mathbf{v}\!(\Delta_\mathcal{T}-1,1) \!+\! \bar{p}_g \bar{p}_0 \bar{p}_1 \pi^{\phi_{\mathcal{T}}}_\mathbf{v}\!(\Delta_\mathcal{T},1), \\
    \pi^{\phi_{\mathcal{T}}}_\mathbf{v}\!(\delta,1) &\!=\! p_g \bar{p}_0 \bar{p}_1 \pi^{\phi_{\mathcal{T}}}_\mathbf{v}\!(\delta\!-\!1,1) \!+\! \bar{p}_g \bar{p}_0 \bar{p}_1 \pi^{\phi_{\mathcal{T}}}_\mathbf{v}\!(\delta,1), \ \hfill \delta \!>\! \Delta_\mathcal{T}.
\end{align*}
\begin{align*}
    \pi^{\phi_{\mathcal{T}}}_\mathbf{v}\!(0,0) & \!=\!\bar{p}_g \pi^{\phi_{\mathcal{T}}}_\mathbf{v}(0,0) \!+\! \bar{p}_g p_1 \pi^{\phi_{\mathcal{T}}}_\mathbf{v}(0,1), \\
    \pi^{\phi_{\mathcal{T}}}_\mathbf{v}\!(\delta,0) &\!=\! \bar{p}_g \pi^{\phi_{\mathcal{T}}}_\mathbf{v}(\delta,0) \!+\! p_g \pi^{\phi_{\mathcal{T}}}_\mathbf{v}(\delta\!-\!1,0)  \\
    &\!+\! p_g p_1 \pi^{\phi_{\mathcal{T}}}_\mathbf{v}(\delta\!-\!1,1) \!+\! \bar{p}_g p_1 \pi^{\phi_{\mathcal{T}}}_\mathbf{v}\!(\delta,1) , \ \hfill 1 \!\leq\! \delta \!<\! \Delta_\mathcal{T}, \\
    \pi^{\phi_{\mathcal{T}}}_\mathbf{v}\!(\Delta_\mathcal{T},0) &\!=\! p_g \pi^{\phi_{\mathcal{T}}}_\mathbf{v}(\Delta_\mathcal{T}-1,0) \!+\! \bar{p}_g \bar{p}_0 \pi^{\phi_{\mathcal{T}}}_\mathbf{v}(\Delta_\mathcal{T},0) \\
    &+\! p_g p_1 \pi^{\phi_{\mathcal{T}}}_\mathbf{v}(\Delta_\mathcal{T}-1,1) \!+\! \bar{p}_g \bar{p}_0 p_1 \pi^{\phi_{\mathcal{T}}}_\mathbf{v}\!(\Delta_\mathcal{T},1), \\
    \pi^{\phi_{\mathcal{T}}}_\mathbf{v}\!(\delta,0) &\!=\! \bar{p}_g \bar{p}_0 \pi^{\phi_{\mathcal{T}}}_\mathbf{v}(\delta,0) \!+\! p_g \bar{p}_0 \pi^{\phi_{\mathcal{T}}}_\mathbf{v}(\delta\!-\!1,0) \\
    &+\! p_g \bar{p}_0 p_1 \pi^{\phi_{\mathcal{T}}}_\mathbf{v}\!(\delta\!-\!1,1) \!+\! \bar{p}_g \bar{p}_0 p_1 \pi^{\phi_{\mathcal{T}}}_\mathbf{v}\!(\delta,1), \ \hfill \delta \!>\! \Delta_\mathcal{T}.
\end{align*}
\begin{align*}
    \pi^{\phi_{\mathcal{T}}}_\mathbf{v}\!(0,1) & \!=\! \frac{\bar{p}_g p_0}{1-\bar{p}_g \bar{p}_1} \bar{\eta}^{\phi_{\mathcal{T}}}_0, \\
    \pi^{\phi_{\mathcal{T}}}_\mathbf{v}\!(1,1) &\!=\! \frac{p_g}{\bar{p}_g(1-\bar{p}_g \bar{p}_1)} \pi^{\phi_{\mathcal{T}}}_\mathbf{v}\!(0,1) = \frac{p_g p_0}{(1-\bar{p}_g \bar{p}_1)^2} \bar{\eta}^{\phi_{\mathcal{T}}}_0, \\
    \pi^{\phi_{\mathcal{T}}}_\mathbf{v}\!(\delta,1) &\!=\! \left[\frac{p_g \bar{p}_1}{1-\bar{p}_g \bar{p}_1}\right]^{\delta\!-\!1} \pi^{\phi_{\mathcal{T}}}_\mathbf{v}(1,1), \quad 2 \leq \delta < \Delta_\mathcal{T},\\
    \pi^{\phi_{\mathcal{T}}}_\mathbf{v}\!(\Delta_\mathcal{T},1) &\!=\! \frac{p_g \bar{p}_1}{1-\bar{p}_g \bar{p}_0 \bar{p}_1} \left[\frac{p_g \bar{p}_1}{1-\bar{p}_g \bar{p}_1}\right]^{\Delta_\mathcal{T}-2} \pi^{\phi_{\mathcal{T}}}_\mathbf{v}(1,1), \\
    \pi^{\phi_{\mathcal{T}}}_\mathbf{v}\!(\delta,1) &\!=\! \frac{p_g \bar{p}_0 \bar{p}_1}{1-\bar{p}_g \bar{p}_0 \bar{p}_1} \pi^{\phi_{\mathcal{T}}}_\mathbf{v}(\delta\!-\!1,1) \\
    &\!=\! \left[\frac{p_g \bar{p}_0 \bar{p}_1}{1-\bar{p}_g \bar{p}_0 \bar{p}_1}\right]^{\delta-\Delta_\mathcal{T}} \pi^{\phi_{\mathcal{T}}}_\mathbf{v}(\Delta_\mathcal{T},1), \quad \delta > \Delta_\mathcal{T},
\end{align*}
which yields:
\begin{align}
    \bar{\eta}^{\phi_{\mathcal{T}}}_1 = \sum_{\delta=0}^{\infty} \pi^{\phi_{\mathcal{T}}}_\mathbf{v}(\delta,1) = \frac{p_0}{p_1} \left(1-\frac{p_0 \Theta^{\Delta_\mathcal{T}}}{1-\bar{p}_0 \bar{p}_1}\right) \bar{\eta}^{\phi_{\mathcal{T}}}_0,
\end{align}
where $\Theta=\frac{p_g \bar{p}_1}{1-\bar{p}_g \bar{p}_1}$.
\end{proof}

    \section{Proof of Proposition \ref{Prop_StateProbRandom}}
	\label{Appen_Proof_RSpolicy_SSProbs}

    \begin{figure}[!tb]
		\centering
		\includegraphics[trim={0.2cm 0cm 0.2cm 0cm}, clip, scale=0.49]{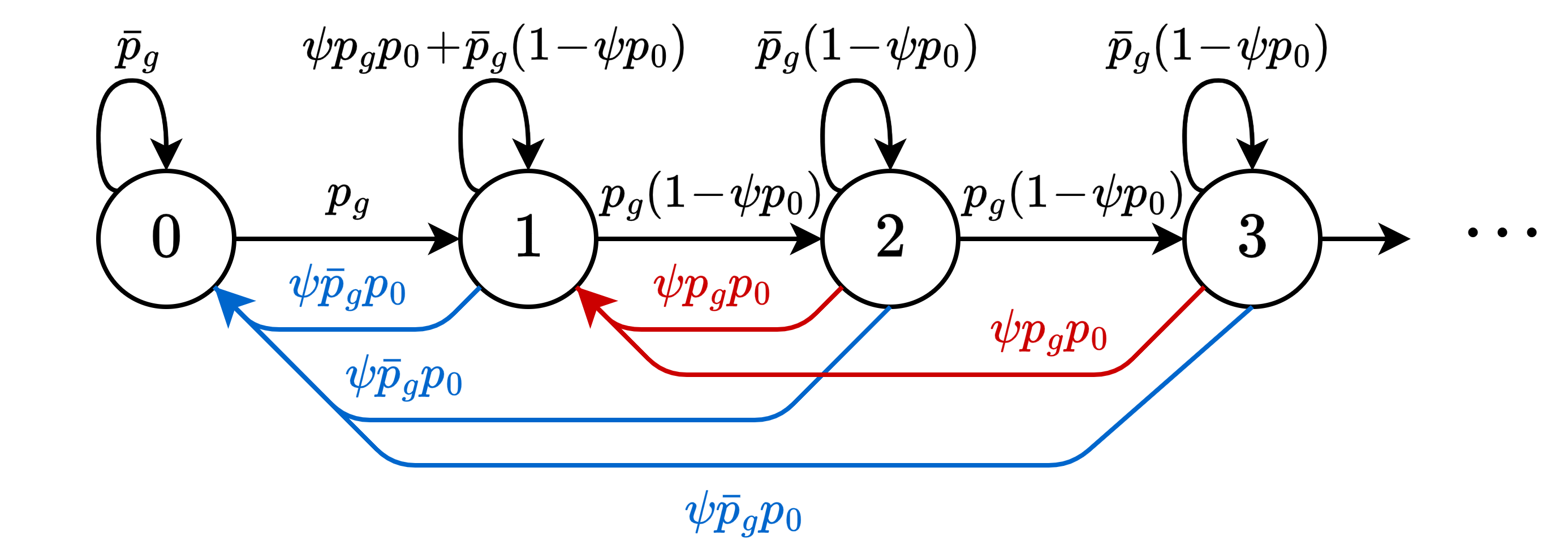} 
		\caption{\textsc{DTMC} model of $\Delta^{\phi_\mathcal{R}}_1$.}
		\label{fig_MC_RSpolicy}
	\end{figure}
    
	\begin{proof}
		At each time slot, a transmission attempt occurs w.p. $\psi$ and succeeds w.p. $p_0$, giving a success probability of $\psi p_0$. Upon success, the \textsc{VAoI} resets to $0$ or $1$ depending on whether a new version is generated in that slot. Otherwise, if no transmission occurs or it fails (w.p. $1 - \psi p_0$), the \textsc{VAoI} increases by $0$ or $1$, again depending on version generation. The Markov chain with transition probabilities is shown in Fig.~\ref{fig_MC_RSpolicy}, and solving its balance equations as described below yields the stationary distribution in Proposition \ref{Prop_StateProbRandom}.
		\[
			\begin{array}{ll}
				\pi^{\phi_{\mathcal{R}}}_{\Delta_1}(0) \!=\! \bar{p}_g \pi^{\phi_{\mathcal{R}}}_{\Delta_1}(0) \!+\! \psi \bar{p}_g p_0 {\sum_{i=1}^{\infty} \! \pi^{\phi_{\mathcal{R}}}_{\Delta_1}(i)}, \\
				\pi^{\phi_{\mathcal{R}}}_{\Delta_1}(1) \!=\! p_g \pi^{\phi_{\mathcal{R}}}_{\Delta_1}(0) \!+\! \big[ \psi p_g p_0 \!+\! \bar{p}_g (1\!-\!\psi p_0) \big] \pi^{\phi_{\mathcal{R}}}_{\Delta_1}(1) \\ 
                \qquad \qquad \quad + \psi p_g p_0 {{\sum_{i=2}^{\infty} \! \pi^{\phi_{\mathcal{R}}}_{\Delta_1}(i)}}, \\
				\pi^{\phi_{\mathcal{R}}}_{\Delta_1}(\delta_1) \!=\! p_g (1\!-\!\psi p_0) \pi^{\phi_{\mathcal{R}}}_{\Delta_1}(\delta_1\!-\!1) \!+\! \bar{p}_g (1\!-\!\psi p_0)\pi^{\phi_{\mathcal{R}}}_{\Delta_1}(\delta_1), \\
                \qquad \qquad \quad \ \ \hfill \delta_1 \geq 2.
			\end{array}
		\]

        The recurrence relation of $\pi^{\phi_{\mathcal{R}}}_{\Delta_1}(\delta_1)$ for $\delta_1 \geq 1$ is geometric: $\pi^{\phi_{\mathcal{R}}}_{\Delta_1}(\delta_1) = r^{\delta_1-1} \pi^{\phi_{\mathcal{R}}}_{\Delta_1}(1)$, where $r = \frac{(1 - \psi p_0)p_g}{\alpha_\mathcal{R}}$. Thus, its expected value is given by $\bar{\Delta}^{\phi_{\mathcal{R}}}_1 = \frac{\pi^{\phi_{\mathcal{R}}}_{\Delta_1}(1)}{(1 - r)^2}$. Noting that $\pi^{\phi_{\mathcal{R}}}_{\Delta_1}(1) = \frac{\psi p_0 p_g}{\alpha_\mathcal{R}^2}$, we obtain $\bar{\Delta}^{\phi_{\mathcal{R}}}_1 = \frac{p_g}{\psi p_0}$.
	\end{proof}

\end{document}